\documentclass[11pt,fleqn]{article}
\usepackage{graphicx,pict2e,amssymb,amsmath,bm,fullpage,bbm,lipsum}
\usepackage{ntheorem,delarray}
\usepackage{hyperref}
\usepackage{srcltx}
\usepackage[margin=1in]{geometry}
\long\def\killtext#1{}
\makeindex\nonstopmode
\theoremseparator{.}
\newtheorem{theorem}{Theorem}[section]

\newtheorem{lemma}[theorem]{Lemma}

\newtheorem{claim}[theorem]{Claim}
\newtheorem{corollary}[theorem]{Corollary}

\theorembodyfont{\rmfamily}

\newtheorem{exa}[theorem]{Example}

\newtheorem{problem}{Problem}[section]

\theoremindent.8cm \theorembodyfont{\small}

\newenvironment{proof}{\noindent{\bf Proof.}}{\hfill$\square$\medskip}

\newenvironment{proof*}[1]{\noindent{\bf Proof of #1.}}{\hfill$\square$\medskip}

\newcommand\numberthis{\addtocounter{equation}{1}\tag{\theequation}}

\newif\ifdraft
\draftfalse

\makeatletter
\g@addto@macro\@floatboxreset\centering
\makeatother

\def\vol{{\sf vol}}
\def\R{{\mathbb R}}
\def\E{{\sf E}}
\def\eps{{\varepsilon}}
\def\Var{{\sf Var}}

\def\ts2a{{t_{\sigma^2/(1-\alpha)}^*}}

\def\dint{{\displaystyle \int}}
\def\k{\frac{1}{e}}
\def\iso{\ln(2)}
\def\slowrate{{\mathchoice{1+\frac{1}{n}}{1+1/n}{1+1/n}{1+1/n}}}
\def\fastrate{{\mathchoice{1+\frac{\sigma^2}{2C^2n}}{1+\sigma^2/(2C^2n)}{1+\sigma^2/(2C^2n)}{1+\sigma^2/(2C^2n)}}}

\newcommand{\e}[1]{
\exp\left({#1}\right)
}

\begin{document}

\title{
%With or Without KLS:
Gaussian Cooling and $O^*(n^3)$ Algorithms\\ for Volume and Gaussian Volume\footnote{This paper combines and extends preliminary conference publications in STOC 2015~\cite{CV2015} and SODA 2014~\cite{CV2014}.}
%Bypassing KLS:\\
%Gaussian Cooling and an $O^*(n^3)$ Volume Algorithm
} 
\author{
Ben Cousins \thanks{Georgia Tech. Supported in part by an NSF graduate fellowship and by NSF award CCF-1217793. Email: {\tt bcousins3@gatech.edu}}
\and
Santosh Vempala\thanks{Georgia Tech. Supported in part by NSF awards CCF-1217793 and EAGER-1555447. Email: {\tt vempala@gatech.edu}}
}
\maketitle

\begin{abstract}
We present an $O^*(n^3)$ randomized algorithm for estimating the volume of a well-rounded convex body given by a membership oracle, improving on the previous best complexity of $O^*(n^4)$.  The new algorithmic ingredient is an accelerated cooling schedule where the rate of cooling increases with the temperature. Previously, the known approach for potentially achieving this asymptotic complexity relied on a positive resolution of the KLS hyperplane conjecture, a central open problem in convex geometry.

We also obtain an $O^*(n^3)$ randomized algorithm for integrating a standard Gaussian distribution over an arbitrary convex set containing the unit ball. Both the volume and Gaussian volume algorithms use an improved algorithm for sampling a Gaussian distribution restricted to a convex body. In this latter setting, as we show, the KLS conjecture holds and for a spherical Gaussian distribution with variance $\sigma^2$, the sampling complexity is $O^*(\max\{n^3, \sigma^2n^2\})$ for the first sample and $O^*(\max\{n^2, \sigma^2n^2\})$ for every subsequent sample. 
\end{abstract}

\section{Introduction}
Computing the volume of a convex body is an ancient and fundamental problem; it is also a {\em difficult} problem, as evidenced by both the $\#$P-hardness of computing the volume of an explicit polytope \cite{DF88} and exponential lower bounds for deterministic algorithms in the general oracle model, even to approximate the volume to within an exponential factor in the dimension \cite{BF87,BF88}. Against this backdrop, the breakthrough result of Dyer, Frieze and Kannan \cite{DyerFK89,DyerFK91} established a randomized polynomial-time algorithm for estimating the volume to within any desired accuracy. In the quarter-century since then, the quest for faster volume algorithms has revealed an array of powerful and elegant techniques for the design and analysis of algorithms, and influenced the development of asymptotic convex geometry \cite{ApplegateK91,LS90,DyerF90,LS93,KLS95,KLS97,BDJ98,LV2,LV06,EldanK2011} . 

The DFK algorithm for computing the volume of a convex body $K$ in $\R^n$ given by a membership oracle uses a sequence of convex bodies $K_0, K_1, \ldots, K_m = K$, starting with the unit ball fully contained in $K$ and ending with $K$. Each successive body $K_i = 2^{i/n}B_n \cap K$ is a slightly larger ball intersected with $K$. Using random sampling, the algorithm estimates the ratios of volumes of consecutive bodies. The product of these ratios times the volume of the unit ball was the estimate of the volume of $K$. Sampling is achieved by a random walk in the convex body. There were many technical  issues to be addressed, but the central challenge was to show a random walk that ``mixed" rapidly, i.e.\ converged to its stationary distribution in a polynomial number of steps. The overall complexity of the algorithm was $O^*(n^{23})$ oracle calls\footnote{The $O^*$ notation suppresses error terms and logarithmic factors.}.

Since then researchers have improved the complexity of volume computation and sampling for convex bodies considerably, to $O^*(n^4)$ for volume estimation and for obtaining the first random sample \cite{LV2,LV06} and to $O^*(n^3)$ per sample for subsequent samples \cite{LV06, LV3}. These improvements rely on continuous random walks, the use of affine transformations, improved isoperimetric inequalities and several other developments. However, throughout the course of these developments, the outer DFK algorithm using a chain of bodies remained unchanged till the most recent improvement in 2003 \cite{LV2}. The LV algorithm \cite{LV2} relies on sampling a sequence of logconcave distributions, akin to simulated annealing, starting with one that is highly concentrated around a point deep inside the convex body and ending with the uniform distribution (we will discuss these ideas in more detail presently). 
The total number of random points needed is only $O^*(n)$, down from $\Omega(n^2)$ needed by all previous algorithms. Combining this with the $O^*(n^3)$ complexity for each sample yielded the overall $O^*(n^4)$ complexity for volume computation. Before running this algorithm, there is a pre-processing step where the convex body is placed in nearly-isotropic position, ensuring in particular that most of the body is contained in a ball of radius $O(\sqrt{n})$. Crucially, this well-roundedness property is maintained during the course of the algorithm.

Is there a faster algorithm? In 1995, Kannan, Lov\'asz and Simonovits, while analyzing the convergence of the ball walk for sampling, proposed a beautiful geometric conjecture now known as the KLS hyperplane conjecture \cite{KLS95}. Roughly speaking, it says that the worst-case isoperimetric ratio for a subset of a convex body is achieved by a hyperplane to within a constant factor. They were able to show that hyperplanes are within $O(\sqrt{n})$ of the minimum. The convergence of the ball walk depends on the square of the reciprocal of the isoperimetric ratio; thus the KLS conjecture had the potential to improve the sampling time by a factor of $n$ to $O^*(n^2)$ per sample and thereby indicated the possibility of an $O^*(n^3)$ volume algorithm (such an algorithm would have to surmount other substantial hurdles). 

The KLS hyperplane conjecture remains unresolved, in spite of intensive efforts and partial progress towards its resolution \cite{Ball88, Fleury2010, EldanK2011}. Indeed, it captures two well-known and much older conjectures from convex geometry, the slicing (or hyperplane) conjecture and the thin-shell conjecture (these were all shown to be equivalent in a certain sense recently \cite{Eldan2013, EldanK2011}), and thus has effectively evaded resolution for nearly a half-century.  

Our first result is an $O^*(n^2)$ algorithm for sampling from the standard Gaussian distribution in $\R^n$ restricted to an arbitrary convex body. To achieve this complexity, we prove the KLS conjecture for such distributions. We then show that the {\em Gaussian volume} or Gaussian measure of a convex body, i.e., the integral of a standard Gaussian over a convex body given by a membership oracle and containing the unit ball. of any convex body, can be computed in $O^*(n^3)$ time.  

Our main finding is an $O^*(n^3)$ algorithm for computing the volume of any convex body containing a unit ball and mostly contained in a ball of radius $O^*(\sqrt{n})$. Equivalently, it suffices to have $\E(\|X\|^2) = O^*(n)$ for a uniform random point $X$ from the body. Assuming the body is well-rounded (or sandwiched) in this sense, no further affine transformation is used, and there is no need to assume or maintain near-isotropy during the course of the volume algorithm.

To describe the main ideas behind the improvement, we recall the LV algorithm in more detail. It uses a sequence of $O^*(\sqrt{n})$ exponential distributions, starting with a distribution that is concentrated inside the unit ball contained in $K$, then ``flattening" this distribution to the uniform by adjusting a multiplicative factor in the exponent\footnote{In the original description, the algorithm first created a ``pencil" using an extra dimension, but this can be avoided \cite{LV06}.}. In each phase, samples from the previous distribution are used to estimate the ratio of the integrals of two consecutive exponential functions (by simply averaging the ratio of the function values at the sample points). It is crucial to keep the variance of this ratio estimator bounded, and to do this, the distributions could be cooled by a factor of $1+\frac{1}{\sqrt{n}}$ in each phase. This leads to $O^*(\sqrt{n})$ phases in total, and to $O^*(\sqrt{n})$ samples per phase. Along with the sample complexity of $O^*(n^3)$ per sample, this gives the bound of $O^*(n^4)$.

%without the KLS conjecture?
%One avenue is indicated by our recent paper \cite{CV2014}, which gives
The improved complexity for Gaussian volume estimation is achieved by using a sequence of Gaussians (rather than exponentials as in LV), starting with a highly concentrated Gaussian centered inside $K$ and ending with the standard Gaussian. The cooling schedule is the same as in the LV algorithm, but each sample takes only $O^*(n^2)$ time. For a Gaussian with covariance $\sigma^2I$, the mixing time is $O^*(\max\{\sigma^2, 1\} n^2)$ (see Theorem \ref{thm:ball-walk} below). Since the starting $\sigma$ is small and the last $\sigma$ is $1$, this bound is $O^*(n^2)$ throughout the algorithm. (We encounter additional technical issues such as maintaining a warm start for the random walks.)

Returning to the usual Lebesgue volume, how could we possibly improve the LV algorithm, without relying on the KLS conjecture?
We will also use Gaussian cooling, starting with a highly concentrated Gaussian and flattening it (i.e., increasing $\sigma$) till we reach the uniform distribution. In the beginning, this is similar to the algorithm of \cite{CV2014}. But after $\sigma$ becomes higher than $1$ (or some constant), we no longer have quadratic sampling time, as the mixing time of the ball work grows as $\max\{\sigma^2, 1\}n^2$. Moreover, we need to go till $\sigma^2 = \Omega(n)$, so cooling at the fixed rate of $1+1/n$ would be too slow. The main new idea is that for $\sigma > 1$, the cooling rate can be made higher, in fact about $1+\sigma^2/n$ instead of only $1+1/n$. This means that the number of phases to double $\sigma^2$ is only $n/\sigma^2$. It can be shown that the number of samples per ``doubling" phase is only $O^*(1)$, giving $n/\sigma^2$ samples in total. Multiplying by the sampling time, we have $\frac{n}{\sigma^2} \cdot \sigma^2 n^2 = n^3$,  a cubic algorithm!
The key technical component of the analysis is to show that the variance of the ratio estimator remains bounded even at this higher cooling rate of $1+\sigma^2/n$.

We now formally state the problems. 

\begin{problem}\label{prob:volume}[Volume]
Given a membership oracle for a convex set $K$ in $\R^n$ containing the unit ball $B_n$, and error parameter $\eps > 0$, give an algorithm that computes a number $V$ such that with probability at least $3/4$,
\[
(1-\eps) \vol(K) \le V \le (1+\eps) \vol(K).
\]
\end{problem}

We denote the Gaussian density function as $\gamma(x) = (2\pi)^{-n/2} \cdot \e{-\|x\|^2/2}$. 

\begin{problem}\label{prob:gauss-volume}[Gaussian Volume]
Given a membership oracle for a convex set $K$ in $\R^n$ containing the unit ball $B_n$, and error parameter $\eps > 0$, give an algorithm that computes a number $V$ such that with probability at least $3/4$,
\[
(1-\eps) \int_K \gamma(x) \, dx \le V \le (1+\eps) \int_K \gamma(x) \, dx.
\]
%$V$ is a $(1+\eps)$ approximation to the Gaussian volume of $K$. 
\end{problem}

\subsection{Main results}

Our main result can be stated more precisely as follows, which solves Problem~\ref{prob:volume} in $O^*(n^3)$ assuming the input body $K$ is well-rounded. We note that the roundness condition can be achieved for any convex body by a preprocessing step consisting of an affine transformation.  It is a significantly weaker condition than isotropic position.

\begin{theorem}\label{thm:volume}
There is an algorithm that, for any $\eps>0, p > 0$ and convex body $K$ in $\R^n$ that contains the unit ball and has $\E_K(\|X\|^2) = O(n)$, with probability $1-p$, approximates the volume of $K$ within relative error $\eps$ and has complexity 
\begin{align*}
&O\left(\frac{n^3}{\eps^2} \cdot \log^2 n \log^2 \frac{1}{\eps} \log^2 \frac{n}{\eps} \log\frac{1}{p}\right) = O^*\left(n^3\right).
\end{align*}
in the membership oracle model. 
\end{theorem}

More generally, if $\E_K(\|X\|^2) = R^2$, then the algorithm has complexity 
\begin{align*}
&O\left(\frac{\max\{R^2n^2, n^3\}}{\eps^2} \cdot \log^2n \log^2 \frac{1}{\eps} \log^2 \frac{n}{\eps} \log\frac{1}{p}\right) 
= O^*\left(\max\{R^2n^2, n^3\}\right).
\end{align*}

The current best complexity for achieving well-roundedness, i.e., $R^2 = O^*(n)$, for a convex body is $O^*(n^4)$ \cite{LV2}. In previous work, the complexity of generating the first nearly uniform random point was always significantly higher than for later points. Here, using a faster cooling schedule, we can generate the first random point in $O^*(n^3)$ steps, under the same assumption that $K$ is well-rounded. Any subsequent uniform random points also require $O^*(n^3)$ steps.

\begin{theorem}\label{thm:sampling}
There is an algorithm that, for any $\eps>0, p>0$, and any convex body $K$ in $\R^n$ that contains the unit ball and has $\E_K(\|X\|^2) = R^2$, with probability $1-p$, generates random points from a density $\nu$ that is within total variation distance $\eps$ from the uniform distribution on $K$. In the membership oracle model, the complexity of each random point, including the first, is 
\begin{align*}
&O\left(\max\{R^2n^2,n^3\}\log n \log^2\frac{n}{\eps} \log \frac{1}{p}\right) = O^*\left(\max\{R^2n^2, n^3\}\right).
\end{align*}
\end{theorem}

In addition to volume and uniform sampling, we also have an $O^*(n^3)$ algorithm for computing the Gaussian volume.  This algorithm does not require a rounding preprocessing step and gives an $O^*(n^3)$ algorithm for any convex set $K\subseteq \R^n$ containing the unit ball. 

\begin{theorem}\label{thm:gauss-volume}
For any $\eps > 0$, $p>0$, and any convex set $K$ in $\R^n$ containing the unit ball, there is an algorithm that, with probability $1-p$, approximates the  
Gaussian volume of $K$ within relative error $\eps$ and has complexity 
\[
O\left(\frac{n^3}{\eps^2}\cdot \log^2\left(n\right) \log^2\left(\frac{n}{\eps}\right)\log\left(\frac{1}{p}\right)\right)=O^*(n^3)
\]
 in the membership oracle model.
\end{theorem}

Both the uniform volume and Gaussian volume algorithms utilize an improved sampling algorithm for Gaussian distributions restricted by convex sets.

\begin{theorem}\label{thm:gauss-sampling}
For any $\eps > 0$, $p>0$, and any convex set $K$ in $\R^n$ containing the unit ball, there is an algorithm that, with probability $1-p$, can generate a random point within total variation distance $\eps$ of the Gaussian density $\mathcal{N}(0,\sigma^2 I)$ restricted to $K$. In the membership oracle model, the complexity of the first random point is 
\[
O\left(\max\{\sigma^2,1\}n^3\log(n)\log^2\left(\frac{n}{\eps}\right)\log\left(\frac{1}{p}\right)\right)=O^*\left(\max\{\sigma^2,1\}n^3\right).
\]
For subsequent random points, the complexity is
\[
O\left(\max\{\sigma^2,1\}n^2\log\left(\frac{n}{\eps}\right)\log\left(\frac{1}{\eps}\right)\right)=O^*\left(\max\{\sigma^2,1\}n^2\right).
\]
 The set of random points will be $\eps$-independent.
\end{theorem}

The following theorem guarantees we can efficiently obtain Gaussian samples from a warm start.

\begin{theorem}\label{thm:ball-walk}
Let $K$ be a convex set containing the unit ball, $Q_0$ be a starting distribution, and $Q$ be the target Gaussian density $\mathcal{N}(0,\sigma^2I)$ restricted to $K \cap 4\sigma\sqrt{n}B_n$. For any $\eps>0,p>0$, the lazy Metropolis ball walk with $\delta$-steps for $\delta = \min\{\sigma,1\}/(4096\sqrt{n\log n/\eps})$, starting from $Q_0$, satisfies $d_{tv}(Q_t,Q) \le \eps$ after
\[
t \ge C \cdot M(Q_0,Q) \cdot \max\{\sigma^2,1\} \cdot n^2 \log\left(\frac{n}{\eps}\right)\log\left(\frac{M(Q_0,Q)}{\eps}\right)
\]
expected steps for an absolute constant $C$.
\end{theorem}

Here $M(Q_0, Q)$ is a measure of how close $Q_0$ is to $Q$ (also called the {\em warm start} parameter) and is defined as 
$M(Q_0,Q) = \sup_{S \subseteq K} \frac{Q_0(S)}{Q(S)}$.
 In other words, the theorem says that ball walk mixes in $O^*(\max\{\sigma^2,1\}n^2)$ steps from a warm start.

\section{Algorithm}

\begin{figure}[h]
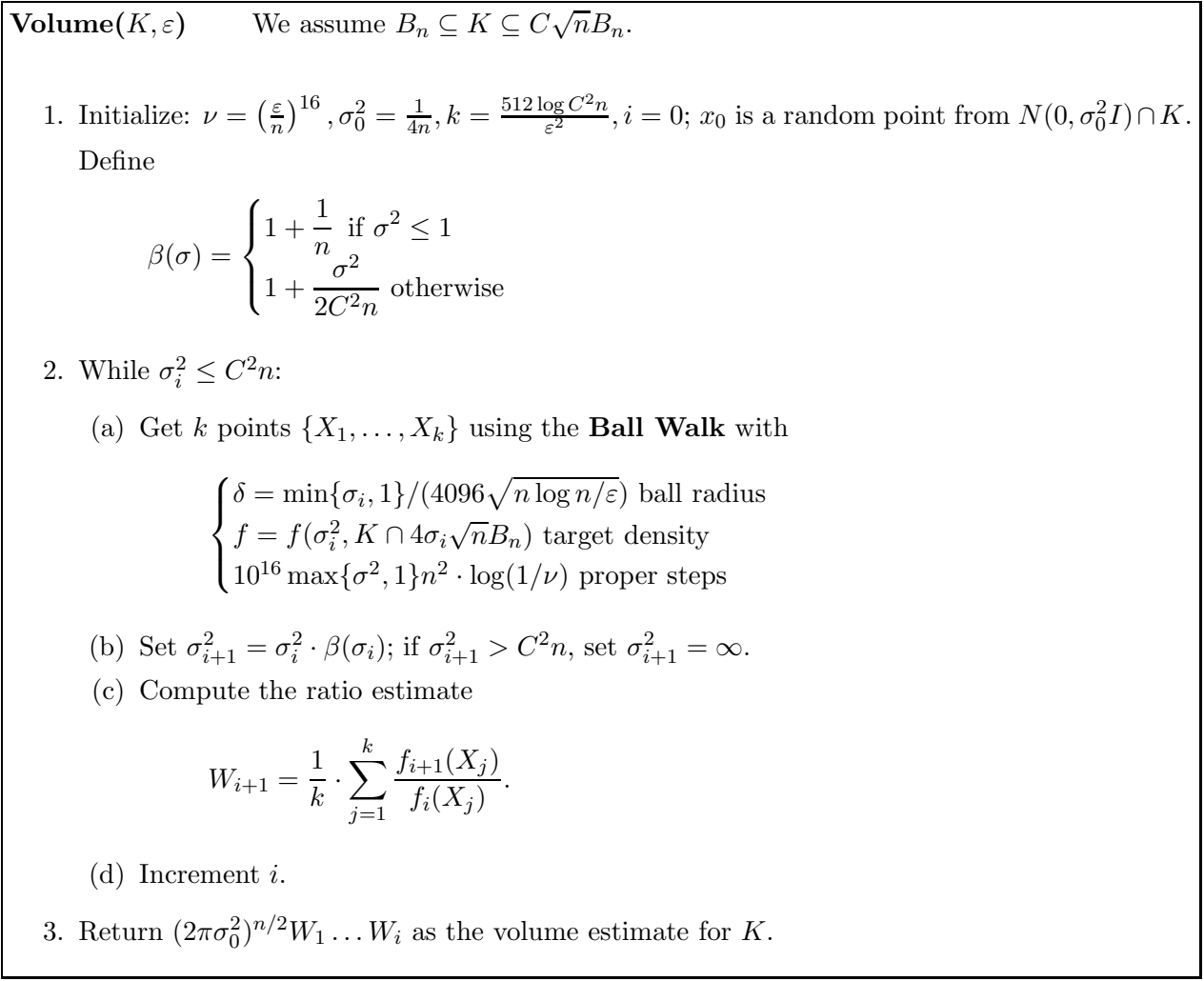

\fbox{\parbox{\textwidth}{
{\bf Volume($K,\eps$)} \qquad We assume $B_n \subseteq K \subseteq C \sqrt{n} B_n$. \\
\begin{enumerate}
\item Initialize: $\nu = \left(\frac{\eps}{n}\right)^{16}, \sigma_0^2 =\frac{1}{4n}, k = \frac{512 \log C^2n}{\eps^2}, i = 0$; $x_0$ is a random point from $N(0,\sigma_0^2 I) \cap K$.

 Define 
\[
\beta(\sigma) = \begin{cases}
\displaystyle\slowrate \, \mbox{ if } \sigma^2 \le 1\\
\displaystyle\fastrate \mbox{ otherwise}
\end{cases}
\]
\item While $\sigma_i^2 \le C^2n$:
\begin{enumerate}
%\item Set $k_\sigma = k\cdot \min\{1/\sigma_i,1\}$.
\item Get $k$ points $\{X_1, \ldots, X_{k}\}$ using the {\bf Ball Walk} with 
\[
\begin{cases}
\delta = \min\{\sigma_i,1\}/(4096\sqrt{n\log n/\eps }) \mbox{ ball radius}\\
f=f(\sigma_i^2, K\cap 4\sigma_i\sqrt{n}B_n) \mbox{ target density}\\
10^{16}\max\{\sigma^2, 1\} n^2 \cdot \log(1/\nu) \mbox{ proper steps}
\end{cases}
\]
%\item Get a {\bf Warm Start} $x_{i+1}$ with convex body $K$, starting point $x_i$ and variance $\sigma_i^2$.

\item Set $\sigma_{i+1}^2= \sigma_i^2 \cdot \beta(\sigma_i)$; if $\sigma_{i+1}^2 > C^2n$, set $\sigma_{i+1}^2 = \infty$.
\item Compute the ratio estimate 
\[
 W_{i+1} = \frac{1}{k} \cdot \sum_{j=1}^{k} \frac{f_{i+1}(X_j)}{f_i(X_j)}.
\]

\item Increment $i$.
\end{enumerate}

\item Return $(2\pi\sigma_0^2)^{n/2}W_1\ldots W_i$ as the volume estimate for $K$.
\end{enumerate}
}}
\caption{The Volume algorithm}\label{fig:volume-alg}
\end{figure}

At a high level, the algorithm relies on sampling random points from a sequence of distributions using the ball walk with a Metropolis filter. For a target density proportional to the function $f$, the ball walk with $\delta$-steps is defined in Figure~\ref{fig:ball-walk}.

\begin{figure}[h]
\fbox{\parbox{4.0in}{
{\bf Ball Walk($\delta, f$)}

At point $x$:
\begin{enumerate}
\item Pick a random point $y$ from $x + \delta B_n$.
\item Go to $y$ with probability $\min\{1, f(y)/f(x)\}$.
\end{enumerate}
}}
\caption{The Ball walk with a Metropolis filter}\label{fig:ball-walk}
\end{figure}

After a suitable number of steps, the point $x$ obtained will be from a distribution close to the one whose density is proportional to $f$. However, this process is slightly complicated by the fact that we only know that the point is mixed once a certain number of \emph{proper} steps have been taken, i.e.\ steps where $y \in K$ or alternatively where $f(y)\neq 0$.

The algorithm in Figure~\ref{fig:volume-alg} starts with a Gaussian of variance $1/(4n)$, with mean at the center of the unit ball inside $K$. This variance is increased over a sequence of phases till the distribution becomes uniform over $K$. Until the variance $\sigma^2$ reaches $1$, it is increased by a fixed factor of $\slowrate$ in each phase. After the variance reaches $1$, the variance accelerates, increasing by a factor of $\fastrate$ where $\sigma^2$ is the current variance. This process is continued till the variance becomes linear in $C^2n$, at which point one final phase can be used to jump to the uniform distribution. In each phase, we pick a sample of random points from the current distribution and compute the average of the ratio of the current density to the next density for each point. The product of these ratios times a fixed term to account for the integral of the initial function is the estimate output by the algorithm.

Let $f(\sigma^2,K)$ be the function that assigns value $\e{-\|x\|^2/(2\sigma^2)}$ to points in a convex set $K$ and zero to points outside. The algorithm below uses a series of such functions. 

\section{Outline of analysis}
\subsection{Outline of sampling analysis}
To show the random walk quickly reaches its stationary distribution, we will use the standard method of bounding the conductance.   
For the ball walk, this runs into a hurdle, namely, the local conductance of points near sharp corners of the body can be arbitrarily small, so the walk can get stuck and waste a large number of steps. To avoid this, we could start the walk from a random point chosen from a distribution sufficiently close to the target distribution. But how to generate random points from such a starting distribution? We do this by considering a sequence of distributions, each providing a warm start for the next. The very first distribution is chosen to be a highly concentrated Gaussian so that it almost entirely lies inside the unit ball (inside $K$). Thus sampling from the initial distribution is easy by standard rejection sampling. Each successive Gaussian is ``flatter" with the final one being the target distribution, e.g.\ standard Gaussian, uniform distribution. 

The next challenge is to show that, from a warm start, the expected number of steps to converge to the stationary distribution is only $O^*(n^2)$. This is usually done by bounding the conductance of the Markov chain. 
The conductance, $\phi$, of a Markov chain with state space $K$ and next-step distribution $P_x$ is defined as:
\[
\phi = \min_{S \subset K} \frac{\int_S P_x(K\setminus S) \, dQ(x)}{\min Q(S), Q(K\setminus S)}.
\]
Unfortunately, for the ball walk, this can be arbitrarily small, e.g., for points near corners (but also for points in the interior). To utilize the warm start, we use an idea from \cite{KLS97}, namely the {\em speedy walk}. We emphasize that the speedy walk cannot be implemented efficiently and is only a tool for analysis. It is defined as follows.

\noindent
At current point $x$:
\begin{enumerate}
\item Pick random point $y$ from $K \cap x + \delta B_n$.
\item Go to $y$ with probability $\min \{1, f(y)/f(x)\}$.
\end{enumerate}

To capture the stationary distribution of the speedy walk with a Metropolis filter we need another parameter. The {\em local conductance} at $x$ for the speedy walk, without a filter, is defined as follows:
\[
\ell(x) = \frac{\vol(K\cap x+\delta B_n)}{\vol(\delta B_n)}.
\]
The following fact is now easy to verify.
\begin{lemma}\label{lem:speedystationary}
The stationary distribution of the speedy walk with a Metropolis filter applied with a function $f$ has density proportional to $\ell(x)f(x)$.
\end{lemma}
For the speedy walk with $\delta = O(1/\sqrt{n})$, we can show that the conductance is $\Omega(1/(\sigma n))$, and so the total number of steps needed is only $O^*(\sigma^2 n^2)$. This is a factor $n$ faster than previous best bounds. We do this by establishing a stronger (and nearly optimal) isoperimetric inequality. 

As noted, the speedy walk cannot actually be implemented efficiently. To bound the Metropolis ball walk, we can view it as an interleaving of a speedy walk with wasted steps. Let the Markov chain for the original walk be $w_0, w_1, \ldots, w_i, \ldots,$. The subsequence $w_{i_1}, w_{i_2}, \ldots,$ where we record $x$ if the point $y$ chosen by the Metropolis ball walk is in $K$, corresponds to the speedy walk. We then need to estimate the number of wasted steps from a warm start. We will show that this is at most a constant factor higher than the number of proper steps. The key ingredient of this analysis is the (known) fact that for a body containing the unit ball average local conductance is high for ball radius $\delta = O(1/\sqrt{n})$. Even within the speedy walk, there are ``null" steps due to the Metropolis filter. However, by restricting the walk to a large ball, we ensure that the probability of rejection by the filter is bounded by a constant, and therefore the number of wasted steps within the speedy walk is at most a constant fraction of all steps. Also, the speedy walk converges to a distribution proportional to $\ell(x)f(x)$, but we can map this to a random sample from $f$ with rejection sampling routine (Section~\ref{sec:speedy-to-target}).

To sample efficiently, we need a warm start for each phase. For two probability distributions $P$ and $Q$ with state space $K$, the $M$-warmness of $P$ and $Q$ is defined as
\begin{equation}
M(P,Q) = \sup_{S \subseteq K} \frac{P(S)}{Q(S)}. \label{eq:M}
\end{equation}
To keep this parameter bounded by a constant,  we use a finer-grained cooling schedule so that a random point from one phase is a warm start for the next phase. This cooling schedule is also different in the two parts. In the first part of the algorithm, where we can cool at the rate of $\slowrate$ and use $O^*(n^2)$ steps to sample. In the second part, we cool at the rate of $\fastrate$, and this is fast enough to compensate for the higher sample complexity of $O^*(\sigma^2n^2)$. Thus, the overall time to obtain a warm start for every phase of the algorithm is also $O^*(n^3)$. We analyze this in full detail in Section~\ref{sec:sampling}, including the proof that this cooling rate maintains a warm start from one phase to the next.

We can obtain uniform random samples from $K$ given samples from a Gaussian with variance $\sigma^2=C^2n$ via a simple rejection sampling routine. Since $K \subseteq C\sqrt{n}B_n$, the two distributions will be within a constant factor of each other, and therefore we can use $O(1)$ expected samples from the Gaussian distribution to obtain a uniform random point.

\subsection{Outline of volume analysis}
The sampling time when the variance is $\sigma^2$ is $\max\{1,\sigma^2\}n^2$. If we cooled at a rate of $\slowrate$ throughout the algorithm, we would get an $O^*(n^4)$ algorithm since the last doubling phase, i.e.\ the set of phases until $\sigma^2$ doubles, takes $\Omega(n)$ samples, each mixing for $\Omega(n^3)$ steps. The main insight that speeds up our algorithm is the cooling rate of $\fastrate$ once $\sigma^2 > 1$. Cooling at a faster rate once $\sigma^2 >1$ will allow us to compute volume in time $O^*(n^3)$ by having fewer phases when the mixing time of the ball walk increases. 

The volume algorithm proceeds as a series of phases, where each phase seeks to estimate a ratio of Gaussian integrals over the convex body $K$. More precisely, let  
\[
f(\sigma^2,x) = \begin{cases}
\e{-\|x\|^2/(2\sigma^2)} \, \mbox{ if } x \in K\\
0 \, \mbox{ otherwise }
\end{cases} \\
\]
and
\[
F(\sigma^2) = \int_{\R^n} f(\sigma^2,x) \, dx.
\]
Define $\mu_i$ as the probability distribution proportional to $f(\sigma_i^2,x)$; that is, $\mu_i$ is a symmetric Gaussian distribution with variance $\sigma_i^2$ restricted to $K$. Let $X$ be a random sample point from $\mu_i$ and let $Y = f(\sigma_{i+1}^2,X)/f(\sigma_i^2,X)$. We see that the expectation of $Y$ is the ratio of $F(\sigma_{i+1}^2)/F(\sigma_i^2)$:
\begin{align*}
\E(Y) &= \int_K \e{\frac{\|x\|^2}{2\sigma_i^2}-\frac{\|x\|^2}{2\sigma_{i+1}^2}}\, d\mu_i(x) \\
&= \int_K \e{\frac{\|x\|^2}{2\sigma_i^2}-\frac{\|x\|^2}{2\sigma_{i+1}^2}} \cdot \frac{\e{-\|x\|^2/(2\sigma_i^2)}}{F(\sigma_i^2)}\, dx \\
&= \frac{1}{F(\sigma_i^2)} \cdot \int_K \e{-\frac{\|x\|^2}{2\sigma_{i+1}^2}}\, dx = \frac{F(\sigma_{i+1}^2)}{F(\sigma_i^2)}.
\end{align*}

Our goal is to estimate $\E(Y)$ within some target relative error. The algorithm estimates the quantity $\E(Y)$ by taking random sample points $X_1, \ldots, X_k$ and computing the empirical estimate for $\E(Y)$ from the corresponding $Y_1, \ldots, Y_k$:
\[
W = \frac{1}{k}\sum_{j=1}^k Y_j = \frac{1}{k} \sum_{j=1}^k \frac{f_{i+1}(X_j)}{f_i(X_j)}. %= \frac{\e{-\frac{\|X_i\|^2}{2\sigma_{i+1}^2}\right)}{\e{-\frac{\|X_i\|^2}{2\sigma_{i}^2}\right)}.
\]

The variance of $Y$ divided by its expectation squared will give a bound on how many independent samples $X_i$ are needed to estimate $\E(Y)$ within the target accuracy. Thus we seek to bound $\E(Y^2)/\E(Y)^2$. We have that
\[
\E(Y^2) = \frac{\int_K \e{\frac{\|x\|^2}{2\sigma_i^2} - \frac{\|x\|^2}{\sigma_{i+1}^2}}\, dx}{\int_K \e{-\frac{\|x\|^2}{2\sigma_i^2}}\, dx} = \frac{F(\frac{\sigma_{i+1}^2\sigma_i^2}{2\sigma_i^2-\sigma_{i+1}^2})}{F(\sigma_i^2)}
\]
and
\[
\frac{\E(Y^2)}{\E(Y)^2} = \frac{F(\sigma_i^2)F(\frac{\sigma_{i+1}^2\sigma_i^2}{2\sigma_i^2-\sigma_{i+1}^2})}{F(\sigma_{i+1}^2)^2} %= \frac{F(\sigma_i^2)F(\sigma_i^2 \cdot (\frac{1}{2\sigma_i^2/\sigma_{i+1}^2 - 1}))}{F(\sigma_{i+1}^2)^2}.
\]

If we let $\sigma^2 = \sigma_{i+1}^2$ and $\sigma_i^2= \sigma^2/(1+\alpha)$, then we can further simplify as
\[
\frac{\E(Y^2)}{\E(Y)^2} = \frac{F\left(\frac{\sigma^2}{1+\alpha}\right)F\left(\frac{\sigma^2}{1-\alpha}\right)}{F\left(\sigma^2\right)^2}.
\]

The algorithm has two parts, and the cooling rate $\alpha_i$ is different for them. In the first part, starting with a Gaussian of variance $\sigma^2=1/(4n)$, which has almost all its measure inside the ball contained in $K$, we increase $\sigma^2$ by a fixed factor of $\slowrate$ in each phase till the variance $\sigma^2$ reaches $1$. For each $\sigma$, we sample random points from the corresponding distribution and estimate the ratio of the densities for the current phase and the next phase by averaging over samples. The total complexity
for the first part is thus
\begin{align*}
&O^*(n) \mbox{ phases} \times O^*(1) \mbox{ samples per phase} \times O^*(n^2) \mbox{ time per sample} = O^*(n^3).
\end{align*}

In the second part, we increase the variance till it reaches $C^2n$, after which one final phase suffices to compare with the target uniform distribution. However, we cannot afford to cool at the same rate of $1+1/n$ because the time per sample goes to $O^*(\sigma^2n^2)$ for $\sigma > 1$. By the end of this part, we would be using $O^*(n^3)$ per sample, and the overall complexity would be $O^*(n^4)$. Instead we observe that we can cool at a faster rate of $\fastrate$ and still maintain that the variance of the ratio estimator is a constant. The following bound on the variance, proved in Section~\ref{sec:variance}, allows us to cool at a faster rate as $\sigma$ increases and overcome the increased sampling cost of $O^*(\sigma^2n^2)$. 

\begin{lemma} \label{lem:variance-bound}
Let $K \subseteq C \sqrt{n} B_n$ and $\alpha \le 1/2$. Then,
\[
 \frac{F\left(\frac{\sigma^2}{1+\alpha}\right)F\left(\frac{\sigma^2}{1-\alpha}\right)}{F\left(\sigma^2\right)} \le \e{2 \cdot \frac{C^2\alpha^2n}{\sigma^2}}.
\]
\end{lemma}

Note that the above RHS is $\le 1+\sigma^2/(Cn)$ if we select $\alpha = \sigma^2/(2C^2n)$. With this rate, the number of phases needed to double the variance is only $O(C^2n/\sigma^2)$, and the number of samples per phase will be $O^*(1)$. Together, they compensate for the higher complexity of obtaining each sample. The complexity of the second part of the algorithm is thus
\begin{align*}
&O^*\left(\frac{C^2n}{\sigma^2}\right) \mbox{ phases} \times O^*\left(1\right) \mbox{ samples per phase} \times O^*(\sigma^2 n^2) \mbox{ time per sample} = O^*(C^2n^3).
\end{align*}

In Section \ref{sec:variance}, we prove that cooling at this accelerated rate still keeps the variance of the ratio estimator appropriately bounded.

We note that with respect to estimating the volume using Lemma~\ref{lem:variance-bound}, there is a range of cooling rates that we could select to obtain an $O^*(n^3)$ algorithm. We need to select $\alpha \le \sigma/(C\sqrt{n})$ to maintain $\alpha\le 1/2$ and satisfy the condition of Lemma~\ref{lem:variance-bound}. We need $\alpha \ge \sigma^2/(C^2n)$ because otherwise there would be too many phases. So for any $\alpha$ such that $c_1 \sigma^2/(C^2n) \le \alpha \le c_2 \sigma/(C\sqrt{n})$, we get that the complexity of the volume algorithm is
\[
O^*\left(\frac{1}{\alpha}\right) \mbox{ phases} \times O^*\left(\frac{C^2\alpha^2 n}{\sigma^2}\cdot \frac{1}{\alpha}\right) \mbox{ samples per phase} \times O^*(\sigma^2 n^2) \mbox{ time per sample} = O^*(C^2n^3).
\]
We select the cooling rate of $\alpha = \sigma^2/(2C^2n)$ for simplicity of the algorithm since this cooling rate also maintains a warm start for the ball walk sampler, as shown in Lemma~\ref{lem:sigma-large-warmness}.

\section{Preliminaries}

A function $f : \R^n \rightarrow \R^+$ is \emph{logconcave} if it has convex support and the logarithm of $f$, wherever $f$ is non-zero, is concave. Equivalently, $f$ is logconcave if for any $x,y \in \R^n$ and any $\lambda \in [0,1]$, 
\[
f(\lambda x + (1  - \lambda) y ) \geq f(x)^\lambda f(y)^{1-\lambda}
\]

Let $\gamma:\R^n \rightarrow \R_+$ be the density of the standard Gaussian $\mathcal{N}(0,I)$.  

For two probability distributions $P$ and $Q$ with state space $K$, we will use $M(P,Q)$ to denote the $M$-warmness between $P$ and $Q$ as defined in \eqref{eq:M}
%\[
%M(P,Q) = \sup_{S \subseteq K} \frac{P(S)}{Q(S)},
%\]
and $d_{tv}(P,Q)$ to denote the total variation distance between $P$ and $Q$:
\[
d_{tv}(P,Q) = \sup_{S \subseteq K} | P(S) - Q(S) | .
\]

For a nonnegative function $f:\R^n \rightarrow \R_+$, we define the $f$-distance between two points $u,v \in \R^n$ as
\[
d_f(u,v) = \frac{|f(u)-f(v)|}{\max\{f(u), f(v)\}}.
\]

\section{Isoperimetry}\label{sec:isoperimetry}

The following theorem is due to Brascamp and Lieb. 
\begin{theorem}\cite{BL76}\label{thm:Brascamp-Lieb}
Let $\gamma:\R^n\rightarrow\R_+$ be the standard Gaussian density in $\R^n$. Let $f:\R^n\rightarrow \R_+$ be any logconcave function. Define the density function $h$ over $\R^n$ as follows:
\[
h(x) = \frac{f(x)\gamma(x)}{\int_{\R^n} f(y)\gamma(y) \, dy}.
\]
Fix a unit vector $v \in \R^n$ ,  let $\mu = \E_h(x)$. Then, for any $\alpha \ge 1$,
\[
\E_h(|v^T(x -\mu)|^\alpha) \le \E_{\gamma}(|x_1|^\alpha).
\]
\end{theorem}
%\vspace*{4.31in}
We have the following concentration bound.

\begin{corollary}\label{cor:iso-conc-bound}
For $h$ as defined in Theorem \ref{thm:Brascamp-Lieb}, and any $t \ge 1$,
\[
\Pr_h(\|x - \mu\|^2  \ge n+ct\sqrt{n}) \le e^{-t^2}
\]
for an absolute constant $c$.
\end{corollary}

%The following theorem is due to Paouris \cite{Paouris2006}.

%\begin{theorem}\label{thm:logconcave-bound}\cite{Paouris2006}
%Let $X$ be drawn from an isotropic logconcave measure. Then, for any $t \geq 1$, 
%\[
%P(\|X\| \geq c t \sqrt{n}) \leq e^{-t\sqrt{n}}
%\]
%where $c$ is an absolute constant.
%\end{theorem}
The next lemma about one-dimensional isoperimetry is from \cite{KLS95}

\begin{lemma}\label{lem:1d-iso}\cite{KLS95}
For any one-dimensional isotropic logconcave function $f$, and any partition $S_1, S_2, S_3$ of the real line, 
\[
\pi_f(S_3) \ge \iso \, d(S_1,S_2) \pi_f(S_1)\pi_f(S_2).
\] 
\end{lemma}

\iffalse
\begin{proof}
By Theorem 5.1 of \cite{KLS95}, 
\[
\pi_f(S_3) \ge \frac{\ln 2}{\E_f(\|x\|)}d(S_1,S_2) \pi_f(S_1)\pi_f(S_2).
\] 
To prove the lemma we note that 
\[
\E_f(\|x\|) \le \sqrt{\E_f(\|x\|^2)} = 1
\]
and 
\[
\pi_f(S_3) \ge c \pi_f(S_1)\pi_f(S_2) \ge c \pi_f(S_1)(1-\pi_f(S_1) - \pi_f(S_3))
\]
which, assuming $\pi_f(S_1) \le \pi_f(S_2)$ and therefore $\pi_f(S_1) \le 1/2$, gives
\[
\pi_f(S_3) \ge \frac{c}{2+c} \min\{ \pi_f(S_1), \pi_f(S_2)\}.
\]
\end{proof}
\fi

\begin{theorem}\label{thm:iso}
Let $\pi$ be the Gaussian distribution $N(0,\sigma^2 I_n)$ with density function $\gamma$ restricted by a logconcave function $f:\R^n \rightarrow \R_+$, i.e., $\pi$ has
density $d\pi(x)$ proportional to $h(x)=f(x)d\gamma(x)$. Let $S_1,S_2,S_3$ partition $\R^n$ such that
for any $u \in S_1, v\in S_2$, either $\|u-v\| \ge d/\iso$ or $d_h(u,v) \ge 4d\sqrt{n}$. Then, 
\[
\pi(S_3) \ge \frac{d}{\sigma}\pi(S_1)\pi(S_2).
\]
\end{theorem}

\begin{proof}
We prove the theorem for the case $\sigma=1$, then note that by applying the scaling $x = y/\sigma$, we get the general case.

Our main tool, as in previous work, is the Localization Lemma of Lov\'asz and Simonovits \cite{LS93}. Suppose the conclusion is false. Define $h(x)=f(x)\gamma(x)$.
 Then there exists a partition $S_1, S_2, S_3$ for which, for some positive real number $A$, 
\begin{align*}
\int_{S_1} h(x)\, dx &=  A\int_{\R^n}h(x)\, dx \\
\mbox{ and } \int_{S_3} h(x)\, dx &<  d A \int_{S_2} h(x)\, dx.
\end{align*}
By the localization lemma, there must be a ``needle'' given by $a,b\in \R^n$ and a nonnegative linear function $l: [0,1] \rightarrow \R_+$ for which,
\begin{align*}
\int_{(1-t)a+tb \in S_1 \cap [0,1]} h((1-t)a+tb) l(t)^{n-1}\, dt = A \int_{(1-t)a+tb \in [0,1]} h((1-t)a+tb) l(t)^{n-1}\, dt
\end{align*}
and
\begin{align*}
\int_{(1-t)a+tb \in S_3 \cap [0,1]} h((1-t)a+tb) l(t)^{n-1}\, dt < d A\int_{(1-t)a+tb \in S_2 \cap [0,1]} h((1-t)a+tb)l(t)^{n-1}\, dt.
\end{align*}
By a standard combinatorial argument, we can assume that $Z_i = \{t: (1-t)a+tb \in S_i\}$ are intervals that partition $[a,b]$.
Thus, to reach a contradiction, it suffices to prove that for a one-dimensional logconcave function $h(t)=f((1-t)a+tb)\gamma((1-t)a+tb)$ with support $[a,b] \subset \R$ and $a\le u\le v \le b$, the
following statements hold:
\begin{align}
\int_{a}^b h(t)l(t)^{n-1}\, dt \int_u^v h(t)l(t)^{n-1}\, dt \ge \frac{d_h(u,v)}{4\sqrt{n}} \int_a^u  h(t)l(t)^{n-1}\, dt  \int_v^b  h(t)l(t)^{n-1} \label{1d-1}
\end{align}
\begin{align}
\int_a^b h(t)l(t)^{n-1}\, dt \int_u^v h(t)l(t)^{n-1}\, dt \ge \iso \|u-v\| \int_a^u  h(t)l(t)^{n-1}\, dt  \int_v^b  h(t)l(t)^{n-1} \label{1d-2}.
\end{align}
The first inequality (\ref{1d-1}) follows directly from Lemma 3.8 in \cite{KLS97}. 
To see the second inequality (\ref{1d-2}), we first note that by applying Theorem \ref{thm:Brascamp-Lieb}, with $\alpha=2$, 
we have that the variance of 
the distribution proportional to $h(t)l(t)^{n-1}$ is at most $1$. This is because $h(t)l(t)^{n-1} = (f((1-t)a+tb)l(t)^{n-1})\gamma((1-t)a+tb)$ and the $f((1-t)a+tb)l(t)^{n-1}$ is itself a logconcave function. Now, we note that by scaling down to increase the variance to exactly $1$, the isoperimetric coefficient can only go down. Hence, the second inequality is implied by Lemma \ref{lem:1d-iso}.
\end{proof}

\section{Sampling}\label{sec:sampling}
The analysis of the sampling algorithm is divided into several parts: bounding the conductance of the speedy walk,  bounding the warmth of the distribution from one phase to the next, the mixing time of the Metropolis ball walk from a warm start, and finally the complexity of sampling.

\subsection{Conductance}\label{sec:conductance}

First we bound the rate of convergence of the random walk through a lower bound on the conductance. The conductance $\phi$ of a Markov chain with state space $K$ and next-step distribution $P_x$ is defined as:
\[
\phi = \min_{S \subset K} \frac{\int_S P_x(K\backslash S)dQ(x)}{\min\{Q(S),Q(K\backslash S)\}}.
\]

We will make use of the following theorem of Lov\'asz and Simonovits \cite{LS93} to bound the total variation distance between the current distribution and the target distribution.

\begin{theorem}\label{thm:mixing}\cite{LS93}
Let $Q_t$ be the distribution after $t$ steps of a lazy Markov chain and $Q$ be its stationary distribution. Suppose that $Q_0$ is $M$-warm with respect to $Q$. Then,
\[
d_{tv}(Q_t,Q) \le \sqrt{M}\left(1-\frac{\phi^2}{2}\right)^t.
\]
\end{theorem}

For the next lemmas, we let $f: \R^n \rightarrow R$ denote the Gaussian density function
\[
f(x) =\e{-\frac{\|x\|^2}{2\sigma^2}}.
\]

\begin{lemma}\label{lem:filter}
For the speedy walk applied to a convex body $K \subseteq 4 \sigma \sqrt{n} B_n$, $\delta \le \sigma/(8\sqrt{n})$ and $\sigma^2 \le 64n$, the acceptance probabilty of the Metropolis filter is at least $\k$. 
\end{lemma}
\begin{proof} Assume that $f(x) \le f(u)$ (otherwise the Metropolis filter always accepts). Then, the acceptance probability is
\begin{align*}
\frac{f(x)}{f(u)} &= \e{-\frac{\|x\|^2 - \|u\|^2}{2\sigma^2}} \\
&\ge \e{-\frac{(\|u\|+\delta)^2 - \|u\|^2}{2\sigma^2}} \\
&= \e{-\frac{2\delta\|u\| + \delta^2}{2\sigma^2}} \\
&\ge \e{-\frac{\sigma^2 + \sigma^2/(64n)}{2\sigma^2}} \\
&\ge \frac{1}{e}.
\end{align*}
\end{proof}

\begin{lemma}\label{lem:f-dist}
Let $K\subseteq 4\sigma\sqrt{n} B_n$ be a convex body and let $u,v \in K$ such that $\|u - v\| \le \delta / \sqrt{n}$. If $\delta \le \sigma/(8\sqrt{n})$ and 
\[
\frac{|\ell(u)f(u) - \ell(v)f(v)|}{\max \{\ell(u)f(u), \ell(v)f(v)\}} < \frac{1}{4},
\]
then
\[
\frac{|\ell(u) - \ell(v)|}{\max \{\ell(u), \ell(v)\}} < \frac{1}{3}.
\]
\end{lemma}

\begin{proof}
Assume without loss of generality that $\ell(u)f(u) \ge \ell(v) f(v)$. It follows that
\[
\frac{3 \gamma(u)\ell(u)}{4} < \gamma(v)\ell(v).
\]

Since $\|u-v\| \le \delta/\sqrt{n}$, we have that $0.9 \le f(v)/f(u) \le 1.1$ for $n$ large enough. Thus
\[
\ell(u) < \frac{3}{2} \ell(v). 
\]

By assumption, we have that $\gamma(u)\ell(u) > \gamma(v)\ell(v)$. Thus we also have that
\[
\ell(v) < \frac{3}{2} \ell(u). 
\]

The lemma then follows. 

\end{proof}

The following lemma bounds the overlap for a step of the speedy walk with respect to the speedy walk. We then show that the Gaussian weighting only hurts by a constant factor.

\begin{lemma}\label{lem:overlap}
Let $K$ be a convex set with $S \subseteq K$. Let $\overline{S} = K \backslash S$ and let $P_x^{unif}$ denote the 1-step distribution from $x$ of the speedy walk with respect to the uniform distribution over $K$. Suppose that $d_\ell(u,v)<1/3$ and $\|u-v\| \le \delta/\sqrt{n}$. Then for $u \in S$ and $v \in \overline{S}$,
\[
P_u^{unif}\left(\overline{S}\right) + P_v^{unif}\left(S\right) > \frac{1}{6}.
\]
\end{lemma}

\begin{proof}
Let $B_u = u + \delta B_n$ and let $C = B_u \cap B_v$. 
By Lemma~3.5 from \cite{KLS95}, we know that
\begin{equation}
\vol\left(K \cap C\right) \ge \frac{\vol(\delta B_n)}{e+1} \min\left\{\ell(u), \ell(v)\right\} \label{eq:kls-overlap}.
\end{equation}

We have that
\[
P_u^{unif}(\overline{S}) = \frac{\vol\left(\overline{S} \cap B_u\right)}{\ell(u) \vol(\delta B_n)} \ge \frac{\vol\left(\overline{S} \cap C\right)}{\ell(u) \vol(\delta B_n)}. 
\]

Similarly for $P_v^{unif}(S)$. Assume that $\ell(u) \ge \ell(v)$, which implies $\ell(u) \le 3\ell(v)/2$. Therefore, 
\begin{align*}
P_u^{unif}(\overline{S}) + P_v^{unif}(S)  &\ge \frac{\vol\left(\overline{S} \cap C\right)}{\ell(u) \vol(\delta B_n)} +\frac{\vol\left(S \cap C\right)}{\ell(v) \vol(\delta B_n)} \\
&\ge \frac{2\vol\left(\overline{S} \cap C\right)}{3\ell(v) \vol(\delta B_n)} +\frac{\vol\left(S \cap C\right)}{\ell(v) \vol(\delta B_n)}  \\
&\ge \frac{2\vol\left(K \cap C\right)}{3\ell(v) \vol(\delta B_n)} \\
&\ge \frac{2}{3(e+1)} > \frac{1}{6}.
\end{align*}

\end{proof}

It then follows that the Gaussian filter decrease the overlap by at most a constant factor.

\begin{corollary}\label{cor:overlap}
Let $S, \overline{S}$ be a partition of a convex body $K \subseteq 4 \sigma \sqrt{n} B_n$, and $u \in S, v \in \overline{S}$ be such that $\|u - v\| < \delta/\sqrt{n}$ and $d_h(u,v)<1/4$, where $h(x) = f(x) \ell(x)$. Then,
\[
P_u(\overline{S}) + P_v(S) > \frac{1}{20}. 
\]
\end{corollary}
\begin{proof}
By Lemma~\ref{lem:f-dist}, we know that $d_\ell(u,v)<1/3$. We then apply Lemma~\ref{lem:overlap}, while noting that the Gaussian weighting affects the $1$-step distributions by at most a $1/e$ factor since that is a lower bound on the acceptance probability of the Metropolis filter (Lemma~\ref{lem:filter}). 
\end{proof}

We can now prove the desired lower bound on the conductance of the speedy walk with respect to a Gaussian weighting over a convex set.

\begin{theorem}\label{thm:speedyconductance}
Let $K$ be a convex body such that $B_n \subseteq K \subseteq 4\sigma \sqrt{n}B_n$.
The conductance of the speedy walk applied to $K$ with Gaussian density ${\cal N}(0,\sigma^2 I)$ and  $\delta \le \sigma/8\sqrt{n}$ steps is $\Omega(\frac{\delta}{\sigma\sqrt{n}})$. 
\end{theorem}
\begin{proof}
Let $S \subset K$ be an arbitrary measurable set of $K$ and let $\overline{S} = K \backslash S$. Assume that $\pi(S) \le 1/2$. Consider the following partition of $K$:
\begin{align*}
S_1 &= \left\{x \in S : P_x(\overline{S}) < \frac{1}{20}\right\} \\ \\
S_2 &= \left\{x \in \overline{S} : P_x(S) < \frac{1}{20}\right\} \\ \\
S_3 &= K \backslash S_1 \backslash S_2.
\end{align*}

Let $h(x) = \ell(x)f(x)$. By Corollary~\ref{cor:overlap}, we have that for any $u \in S_1, v \in S_2$, either $\|u-v\| \ge \delta/\sqrt{n}$ or $d_h(u,v) \ge 1/4$. 

We may assume that $\pi(S_1) \ge \pi(S)/2$ and $\pi(S_2) \ge \pi(\overline{S})/2$. If not, we can bound the conductance of $S$ as follows (similarly for $\overline{S}$).
\begin{align*}
\phi(S) &= \frac{\dint_{S} P_x(\overline{S})h(x)\, dx}{\pi(S)} \\
&=\frac{1}{2} \frac{\dint_{S} P_x(\overline{S})h(x)\, dx + \dint_{\overline{S}} P_x(S)h(x)\, dx}{\pi(S)} \\
&\ge \frac{1}{2} \frac{\dint_{S_3} \dfrac{h(x)}{20} \, dx }{\pi(S)} \\
&= \frac{1}{40} \frac{\pi(S_3)}{\pi(S)} \\
&\ge \frac{1}{80}. 
\end{align*}

Now we can apply Theorem~\ref{thm:iso} with 
\[
d = \min\left\{\frac{\delta \ln 2}{\sqrt{n}}, \frac{1}{16\sqrt{n}}\right\}
\]
to the partition $S_1, S_2, S_3$ to get
\[
\pi(S_3) \ge \frac{d}{\sigma} \pi(S_1) \pi(S_2). 
\]

Using the above, we get that 
\begin{align*}
\phi(S) &\ge \frac{1}{40} \frac{\pi(S_3)}{\pi(S_1)} \\
& \ge \frac{d}{40\sigma} \frac{\pi(S_1)\pi(S_2)}{\pi(S_1)} \\
& \ge \frac{d}{160\sigma} \\
&\ge \frac{\delta}{250 \sigma \sqrt{n}},
\end{align*}

which proves the theorem.
\end{proof}

\subsection{Getting a warm start}\label{sec:warmstart}

The following two lemmas guarantee that the ball walk in the algorithm will always have a \emph{warm start}, i.e.\ the $M$-warmness \eqref{eq:M} is bounded by a constant. The first lemma bounds the warmness under the fixed cooling rate of $1 + 1/n$.

\begin{lemma}\label{lem:sigma-small-warmness}
Let $K\subseteq \R^n$, $\sigma_{i+1}^2 = \sigma_i^2(1+1/n)$, and $f_i(x) = \e{-\|x\|^2/(2\sigma_i^2)}$. Denote $Q_i$ as the associated probability distribution of $f_i$ over $K$. Then, we can bound the warmness between successive
\[
M(Q_i,Q_{i+1}) \leq \sqrt{e}.
\]
\end{lemma}

The following lemma bounds the warmness when the cooling schedule begins to accelerate, under the roundness condition.

\begin{lemma}\label{lem:sigma-large-warmness}
Let $K \subseteq C \sqrt{n} \cdot B_n$, $\sigma_{i+1}^2 = \sigma_i^2 (1+\sigma_i^2/(C^2n))$, and $f_i(x)=\e{-\|x\|^2/(2\sigma_i^2)}$. Denote $Q_i$ as the associated probability distribution of $f_i$ over $K$. Then we can bound the warmness between successive phases as 
\[
M(Q_i,Q_{i+1}) \le \sqrt{e}.
\]
\end{lemma}

\begin{proof}(of Lemma~\ref{lem:sigma-small-warmness})
Let
\[
A = \frac{\int_K e^{-a_{i+1}\|x\|^2}dx}{\int_K e^{-a_i \|x\|^2}dx}.
\]
Then, 
\begin{align*}
M(Q_i,Q_{i+1}) &= \sup_{S \subseteq K} \frac{Q_i(S)}{Q_{i+1}(S)} \\
&\leq \sup_{x \in K} \frac{Q_i(x)}{Q_{i+1}(x)} \\
&= \sup_{x \in K} A \frac{e^{-a_i \|x\|^2}}{e^{-a_{i+1} \|x\|^2}} \\
&= A\cdot \sup_{x \in K} e^{-a_i\|x\|^2/n}\\
&=A,
\end{align*}
where the last line follows from the fact that $0 \in K$.

We will now bound $A$. First, we extend $A$ to be over all $\R^n$ instead of $K$, and then argue that it can only decrease when restricted to $K$.
\begin{align*}
\frac{\int_{\R^n} e^{-a_{i+1}\|x\|^2}dx}{\int_{\R^n}e^{-a_i \|x\|^2}dx} &=
 \frac{(a_{i}/\pi)^{n/2}}{(a_{i+1}/\pi)^{n/2}}\frac{(a_{i+1}/\pi)^{n/2}}{(a_{i}/\pi)^{n/2}}\frac{\int_{\R^n} e^{-a_{i+1}\|x\|^2}dx}{\int_{\R^n}e^{-a_i \|x\|^2}dx}\\
&=\frac{a_i^{n/2}}{a_{i+1}^{n/2}} \\
&= \left(\frac{1}{1-1/n}\right)^{n/2} \\
&\leq \sqrt{e}.
\end{align*}

Let $\mu_K(r)$ be the proportion of the sphere of radius $r$ centered at $0$ that is contained in $K$. Note that since $K$ is a convex body that contains $0$, that $r_1 > r_2 \Rightarrow \mu_{K}(r_1) \leq \mu_{K}(r_2)$. Then,
\[
A = \frac{\int_0^\infty r^{n-1} e^{-a_{i+1} r^2} \mu_K(r) dr}{\int_0^\infty r^{n-1} e^{-a_i r^2} \mu_K(r) dr}.
\]

Note $(r^{n-1}e^{-a_{i+1}r^2})/(r^{n-1}e^{-a_i r^2})$ is a monotonically increasing function in $r$. Since $K$ is a convex body containing $0$, we can partition $K$ into infinitesimally small cones centered at $0$. Consider an arbitrary cone $C$. $\mu_C(r)$ is $1$ for $r \in [0,r']$ and then $0$ for $r \in (r',\infty)$ since $K$ is convex. Since $(r^{n-1}e^{-a_{i+1}r^2})/(r^{n-1}e^{-a_i r^2})$ is monotonically increasing, the integral over the cone only gets larger by extending $\mu_C(r)$ to be $1$ for $r \in [0, \infty)$. Therefore
\begin{align*}
\frac{\int_0^\infty r^{n-1} e^{-a_{i+1} r^2} \mu_C(r) dr}{\int_0^\infty r^{n-1} e^{-a_i r^2} \mu_C(r) dr} \leq \frac{\int_0^\infty r^{n-1} e^{-a_{i+1} r^2}dr}{\int_0^\infty r^{n-1} e^{-a_i r^2}dr} = \sqrt{e}.
\end{align*}
Since $C$ was an arbitrary cone from a partition of $A$, we have that $A \leq \sqrt{e}$.
\end{proof}

\begin{proof}(of Lemma \ref{lem:sigma-large-warmness}).
Note that
\begin{align*}
f_{i+1}(x) &= \e{-\frac{\|x\|^2}{2\sigma_{i+1}^2}} \\
&= \e{-\frac{\|x\|^2}{2\sigma_i^2(1+\sigma_i^2/(C^2n))}} \\
&= \e{-\frac{\|x\|^2}{2\sigma_i^2}\cdot \left(1-\frac{\sigma_i^2/(C^2n)}{1+\sigma_i^2/(C^2n)}\right)} \\
&= f_i(x) \cdot \e{\frac{\|x\|^2}{2C^2n} \cdot \frac{1}{1+\sigma_i^2/(C^2n)}}\\
&\le f_i(x) \cdot \e{\frac{\|x\|^2}{2C^2n}}.
\end{align*}

We have that
\begin{align*}
M(Q_i,Q_{i+1}) &= \sup_{S \subseteq K} \frac{Q_i(S)}{Q_{i+1}(S)} \\ \\
&\le \frac{\int_K f_{i+1}(x)\, dx}{\int_K f_i(x)\, dx} \cdot \sup_{x \in K} \frac{f_i(x)}{f_{i+1}(x)} \\ \\
&= \frac{\int_K f_{i+1}(x)\, dx}{\int_K f_i(x)\, dx}  \cdot \sup_{x \in K} \left(\e{-\frac{\|x\|^2}{2C^2n} \cdot \frac{1}{1+\sigma_i^2/(C^2n)}}\right) \\ \\
&= \frac{\int_K f_{i+1}(x)\, dx}{\int_K f_i(x)\, dx} \\ \\
&\le \frac{\int_K f_i(x)\e{\|x\|^2/(2C^2n)}\, dx}{\int_K f_i(x)\, dx} \\ \\
&\le \sup_{x \in K} \left(\e{\frac{\|x\|^2}{2C^2n}}\right) \\ \\
&\le \sqrt{e}
\end{align*}
since $\|x\| \le C \sqrt{n}$.
\end{proof}

\subsection{Bounding wasted steps}

The speedy walk is defined as the \emph{proper} steps of the ball walk, where the point the ball walk attempts to visit is contained in $K$. For convenience, we restate the definition of the speedy walk from earlier (Figure~\ref{fig:speedy-walk}).

\begin{figure}[h!]
\fbox{\parbox{4.0in}{
{\bf Speedy Walk($\delta, f$)}\\
At current point $x \in K$:
\begin{enumerate}
\item Pick random point $y$ from $K \cap (x + \delta B_n)$.
\item Go to $y$ with probability $\min\{1,f(y)/f(x)\}$.
\end{enumerate}
}}
\caption{The Speedy walk with a Metropolis filter}\label{fig:speedy-walk}
\end{figure}

To prove convergence of the ball with a Metropolis filter, we prove convergence of the speedy walk, then bound the number of ``wasted" steps. Note that the speedy walk cannot be implemented as described in Figure~\ref{fig:speedy-walk}, but is an analysis tool to prove the mixing time of the ball walk.

Next, we bound the average number of wasted steps of the ball walk, i.e., when the ball walk tries to visit a point not in $K$. The average local conductance of the ball walk is defined as 
\[
\lambda(f) = \frac{\int_K \ell(x)f(x)\, dx}{\int_K f(x)\, dx}.
\]

We say that a density function $f: \R^n \rightarrow \R_+$ is $a$-\emph{rounded} if any level set $L$ contains a ball of radius $a \cdot \mu_f(L)$. We now show that the average local conductance is large, i.e.\ at least a constant.

\begin{lemma}\label{lem:lambda-bound}
For any $a$-rounded logconcave density function $f$ in $\R^n$,
\[
\lambda(f) \ge 1 - 32 \frac{\delta^{1/2}n^{1/4}}{a^{1/2}}.
\]
\end{lemma}

\begin{proof}
Define $\hat{f}$ as the following smoothened version of $f$, obtained by convolving $f$ with a ball of radius $\delta$. 
Let $D$ be a convex subset of $\delta B_n$ of half its volume.
\[
\hat{f}(x) = \min_{D} \frac{\int_{y \in x+D} f(y)\, dy}{\vol(D)}.
\]
Now Lemma 6.3 from \cite{Lovasz2007} shows that
\[
\int_K \hat{f}(x)\, dx \ge 1 - 32 \frac{\delta^{1/2}n^{1/4}}{a^{1/2}}.
\]
To complete the proof, we observe that for any point $x$,
\[
\ell(x)f(x) \ge \hat{f}(x).
\]
To see this, note that 
\begin{eqnarray*}
\ell(x)f(x) &=& f(x) \frac{\int_{x+\delta B_n} 1\, dy}{\vol(\delta B_n)}\\
&\ge& \frac{\int_{y \in x+\delta B_n: f(y) \le f(x)} f(y)\, dy}{\int_{y \in x+\delta B_n:f(y) \le f(x)} 1 \, dy}\\
&\ge& \hat{f}(x).
\end{eqnarray*}
\end{proof}

\begin{lemma}\label{lem:sigma-roundness}
The Gaussian $\mathcal{N}(0,\sigma^2 I)$ restricted to $K$ containing a unit ball centered at $0$ is $\min\{\sigma,1\}$-rounded.
\end{lemma}
\begin{proof}
The level sets of the distribution are balls restricted to $K$. For the distribution to be $\min\{\sigma,1\}$-rounded, we need that a level set of measure $k$ contains a ball of radius $k \cdot \min\{\sigma,1\}$. Consider the following function of $t$, which is an upper bound on the measure of the ball of radius $t \le \min\{\sigma, 1\}$, since the unit ball is contained in $K$:
\[
g(t) = \frac{\int_0^t x^{n-1} \e{-\frac{x^2}{2\sigma^2}}\, dx}{\int_0^{\min\{\sigma,1\}} x^{n-1} \e{-\frac{x^2}{2\sigma^2}}\, dx}.
\]

Consider the second derivative of $g$: 

\[
g''(t) = \left((n-1) - \frac{t^2}{\sigma^2}\right) \cdot \frac{t^{n-2} \e{-\frac{t^2}{2\sigma^2}}}{\int_0^{\min\{\sigma,1\}} x^{n-1} \e{-\frac{x^2}{2\sigma^2}}\, dx}.
\]

For $g''(t)$ to be nonnegative, we need $\sigma^2(n-1) - t^2 \ge 0$, which it is for $n \ge 2, t \in [0,\min\{\sigma,1\}]$. Since $g(0) = 0$, $g(\min\{\sigma,1\}) = 1$, and the second derivative is nonnegative, we then have that $g(t\min\{\sigma,1\}) \le t$ for $t \in [0,1]$, which proves the lemma.

\end{proof}

We now show that for an appropriate selection of ball radius, the ball walk has large average local conductance.

\begin{lemma}
If $\delta \le \min\{\sigma,1\}/(4096\sqrt{n})$, then the average local conductance, $\lambda(f)$, for the density function $f$ proportional to the Gaussian $\mathcal{N}(0,\sigma^2I_n)$ restricted to $K$ containing the unit ball, is at least $1/2$.
\end{lemma}
\begin{proof}
Using Lemma~\ref{lem:lambda-bound} and Lemma~\ref{lem:sigma-roundness}, we have that 
\[
\lambda(f) \ge 1 - 32 \frac{\min\{\sigma^{1/2},1\} n^{1/4}}{64n^{1/4}\min\{\sigma^{1/2},1\}} = \frac{1}{2}.
\]
\end{proof}

The following lemma is shown in \cite{CV2014}.

\begin{lemma}\label{lem:speedy-to-ball}
If the average local conductance is at least $\lambda$, $M(Q_0,Q) \le M$, and the speedy walk takes $t$ steps, then the expected number of steps of the corresponding ball walk is at most $Mt/\lambda$.
\end{lemma}

\begin{proof}
Since $M(Q_0,Q)\leq M$, we have that for all $S \subseteq K$,
\[
Q_0(S) \leq M Q(S),
\]

and by induction on $i$, we get that
\begin{align*}
Q_i(S) = \int_K P_x(S) dQ_{i-1}(x) \leq M\int_K P_x(S) dQ(x) = MQ(S).
\end{align*}

For any point $x$, the expected number of steps until a proper step is made is $1/\ell(x)$. So, given a point from $Q_i$, the expected number of steps to obtain a point from $Q_{i+1}$ is 
\begin{align*}
\int_K \frac{1}{\ell(x)}\, dQ_i(x) \leq M\int_K \frac{1}{\ell(x)}\, dQ(x) = M \int_K \frac{1}{\lambda}\,  d\hat Q (x) = \frac{M}{\lambda},
\end{align*}

where $\hat Q$ is the corresponding distribution for the ball walk with a Metropolis filter (i.e., with stationary distribution proportional to $f(x)$). If the speedy walk took $t$ steps, then by linearity of expectation, the expected number of steps for the ball walk is at most $Mt/\lambda$.
\end{proof}

\subsection{Mapping speedy distribution to target distribution}\label{sec:speedy-to-target}

When the speedy walk has converged, we obtain a point approximately from the speedy walk distribution $\ell(x) f(x)$. We will use a rejection routine to map a random point from this distribution to the target distribution $f(x)$ while incurring a small amount of additional sampling error.  We adapt the proof of Theorem~4.16 of \cite{KLS95} to the Gaussian setting.

\begin{lemma}\label{lem:map-speedy-to-ball}
Assume that $\|P-\hat{Q}\|_{tv} \le \eps$, $B_n \subseteq K$, $\eps \le 1/10$, and 
\[
\delta \le \frac{\min\{\sigma,1\}}{8\sqrt{n\log(n/\eps)}}.
\]
There is an algorithm that will use a constant number of random samples from $P$, in expectation, to obtain a distribution $R$ satisfying $\|R-Q\|_{tv} \le 10\eps$. 
\end{lemma}
\begin{proof}
The rejection routine is as follows: let $c=1-1/(2n)$. For a point $u$ from distribution $\ell(x) f(x)$, let $v=(1/c)u$. Accept $v$ with probability $f(v)/f(u)$. Repeat until we accept a $v$. 

The correctness of the above routine follows from the following two facts: (i) with constant probability, the rejection sampling will succeed and (ii) removing a thin shell around the boundary makes $\ell(x)$ look close to uniform on average.

Recall that $\hat{Q}$ is the speedy walk distribution and $Q$ is the ball walk distribution. Consider a level set $\mu_L = \{x : f(x) \ge L\}$. By logconcavity of $f$, $\mu_L$ is convex. From \cite{KLS95}, $\hat{Q}_{\mu_L}(c \cdot \mu_L) \ge 1/2$. By applying this to all level sets, it then follows that $\hat{Q}(cK) \ge 1/2$. 

Also from \cite{KLS95}, if $\mu_L$ contains the unit ball, then 
\[
\dint_{\mu_L\cap cK} \left(1-\ell(x)\right)\, dx \le \eps \vol\left(\mu_L \cap cK\right).
\]
Recall that the level sets $\mu_L$ are balls intersected with $K$ since $f$ is a spherical Gaussian distribution. If $\mu_L$ does not contain the unit ball, a standard calculation (using that $B_n \subseteq K$) shows that the local conductance is at least $1-\eps$ for every point, and thus
\[
\dint_{\mu_L\cap cK} (1-\ell(x))\, dx \le \eps \vol(\mu_L \cap cK).
\]

Using the above, we see that
\begin{align*}
\dint_{cK} \left(1-\ell(x)\right)f(x)\, dx &= \dint_0^\infty \dint_{\mu_L} \left(1-\ell(x)\right)\{x \in \mathbbm{1}_{cK}\}\, dx\, dL \\
&\le \dint_0^\infty \eps \vol(\mu_L \cap cK)\, dL \\
&\le \eps \dint_{cK}f(x)\, dx.
\end{align*}

Then, 
\begin{align*}
\hat{Q}(cK) &= \frac{\dint_{cK} \ell(x) f(x)\, dx}{\dint_K \ell(x) f(x) \, dx} \\
&= \frac{\dint_{cK}f(x)\, dx - \dint_{cK} \left(1-\ell(x)\right)f(x)\, dx}{\dint_K \ell(x) f(x)\, dx} \\
&\ge \frac{\dint_{cK} f(x)\, dx - \eps \dint_{cK} f(x)\, dx}{\dint_K \ell(x)f(x)\, dx}
\end{align*}
and 
\[
\hat{Q}(cS) \le \frac{\dint_{cS} f(x)\, dx}{\dint_K \ell(x)f(x)\, dx}. 
\]

Let $P'$ be the distribution of the first sample from $P$ which satisfies $(1/c)x\in K$. Define
\[
z(x) = \begin{cases}
f(cx) \mbox{ if } x \in K \\
0 \mbox{ otherwise}
\end{cases}
\]
and let $Z$ be the probability distribution corresponding to $z$. Then,
\begin{align*}
P'(S)-Z(S) &= \frac{P(cS)}{P(cK)} - \frac{Q(cS)}{Q(cK)} \\
&\le \frac{\hat{Q}(cS) +\eps}{\hat{Q}(cK)-\eps} - \frac{Q(cS)}{Q(cK)} \\
&\le \frac{\dint_{cS} f + \eps \dint_K \ell(x)f(x)\, dx }{\dint_{cK} f - \eps \dint_{cK} f - \eps \dint_K \ell(x)f(x)\, dx } - \frac{\dint_{cS} f(x)\, dx}{\dint_{cK} f(x)\, dx} \\ 
&\le \frac{1+2\eps}{1-2\eps} - 1 \\
&\le 10\eps.
\end{align*}

Then accept a point $x$ with probability $f(x)/z(x)$, which is at least a constant since $\|x\| \le 4\sigma\sqrt{n}$. The overall expected number of rejection steps is a constant since $\hat{Q}(cK) \ge 1/2$.
\end{proof}

\subsection{Proof of sampling theorems}

We can now prove Theorem \ref{thm:gauss-sampling} and Theorem~\ref{thm:ball-walk} for sampling a Gaussian distribution restricted to a convex body.\\

\begin{proof}(of Theorem~\ref{thm:ball-walk})

By Theorem~\ref{thm:speedyconductance} and Theorem~\ref{thm:mixing}, we have selecting $\delta = \min\{\sigma,1\}/(4096\sqrt{n})$ implies that the speedy walk starting from a distribution that is $M$-warm will be within total variation distance $\eps$ of the target distribution in 
$O(\max\{\sigma^2,1\}n^2\log(n/\eps)\log(M/\eps))$
steps. 

By Lemma~\ref{lem:speedy-to-ball}, the ball walk will, in expectation, take at most $2M$ times as many steps since the average local conductance $\lambda$ is at least $1/2$. Therefore, the total number of expected ball walk steps is $O(M \max\{\sigma^2,1\}n^2\log(n/\eps)\log(M/\eps))$. We then repeat this walk $O(1)$ times until we obtain a point from the proper target distribution using Lemma~\ref{lem:map-speedy-to-ball}.

\end{proof}

\begin{proof}(of Theorem \ref{thm:gauss-sampling}.)

Note that here, we are analyzing the sampling phases of Figure~\ref{fig:volume-alg}, and only the phases when $\sigma^2\le 1$. 

By Theorem~\ref{thm:ball-walk}, we have that the ball walk will take $O(M \max\{\sigma^2,1\} n^2 \log(n/\eps)\log(M/\eps))$ steps in expectation. By Lemma~\ref{lem:sigma-small-warmness}, each phase will always provide a warm start to the next, i.e. $M=O(1)$. By assigning a sampling error $(\eps/n)^{16}$ to each phase, we ensure that the overall sampling failure is at most $\eps$ by a straightforward union bound. Therefore, each sampling phase takes
\[
O\left(n^2 \log^2\left(\frac{n}{\eps}\right)\right)
\]
expected steps of the ball walk. Adding up across phases introduces an additional $n \log n$ factor since we increase $\sigma^2$ by the rate of $\slowrate$ between phases.

If we want to instead run for a fixed number of steps, we can keep a global counter of the ball walk steps. Say the expected number of ball walk steps is $T$. If at any point the number of ball walk steps goes above $2T$, we abandon this run of the algorithm. The probability of a single run failing is at most $1/2$ by Markov's inequality. If we want an overall failure probability of at most $p$, then we can run $\log(1/p)$ iterations of the algorithm, and with probability $1-p$, at least one of them will succeed.
\end{proof}

\begin{proof}(of Theorem~\ref{thm:sampling})
The proof of Theorem~\ref{thm:sampling}, which extends Gaussian sampling to uniform sampling, follows along the same lines as Theorem~\ref{thm:gauss-sampling}. When $\sigma^2\le 1$, the total expected ball walk steps is
\[
O\left(n^3 \log(n) \log^2\left(\frac{n}{\eps}\right)\right).
\]

When $\sigma^2 > 1$, we additionally use Lemma~\ref{lem:sigma-large-warmness}, which implies that we can accelerate our cooling rate and still maintain a warm start. This accelerated rate allows us to overcome the increased mixing time of $O^*(\max\{\sigma^2,1\}n^2)$ once $\sigma^2\ge 1$. Now consider a ``chunk" of phases as a set of phases until $\sigma^2$ doubles. There will be $O(C^2n/\sigma^2)$ phases in a chunk, where each chunk has expected mixing time $O(\sigma^2n^2 \log(n/\eps))$. Since there are $O(\log n)$ chunks (provided $C=\text{poly}(n)$), the total number of expected ball walk steps when $\sigma^2 > 1$ is 
\[
O\left(C^2n^2 \log(n) \log\left(\frac{n}{\eps}\right)\right).
\]

Note that this will yield a random sample with respect to a Gaussian with $\sigma^2=C^2n$ restricted to $K$. We can map this point to a uniform random point using simple rejection sampling, which will succeed with probability at least $1/e$ since $K \subseteq C\sqrt{n}$. If it fails, we can restart the algorithm. As with Theorem~\ref{thm:gauss-sampling}, we can repeat $\log(1/p)$ times to transform the expected ball walk steps into a fixed number of steps with success probability $1-p$. 

\end{proof}

\section{Analysis of volume algorithm}
\subsection{Accelerated cooling schedule}\label{sec:variance}
%
%\ifdraft
%\stepcounter{subsection}
%\setcounter{theorem}{9}
%\else

The goal of this section is to prove the Lemma \ref{lem:variance-bound}, which gives a bound on the variance of the random variable we use to estimate the ratio of Gaussian integrals in the volume algorithm in Figure~\ref{fig:volume-alg}. Here we will actually prove the inequality to be true for all logconcave functions, but only apply it to an indicator function of a convex body.
Let $f: \R^n \rightarrow \R$ be a logconcave function such that $E_f(\|X\|^2) = R^2$.

Define 
\[
g(x,\sigma^2) = f(x) \e{-\frac{\|x\|^2}{2\sigma^2}}
\]
and also define
\[
G(\sigma^2) = \dint_{\R^n} g(x,\sigma^2)\, dx. 
\]

Define $\mu_i$ as the probability distribution proportional to $g(x,\sigma_i^2)$. Let $X$ be a random sample from $\mu_i$ and let $Y=g(X,\sigma_{i+1}^2)/g(X,\sigma_i^2)$. From a standard calculation, we have that
\[
\E(Y) = \frac{G(\sigma_{i+1}^2)}{G(\sigma_i^2)}.
\]
The second moment of $Y$ is given by
\begin{align*}
\E(Y^2) &= \dint_{\R^n} \left(\frac{g(x,\sigma_{i+1}^2)}{g(x,\sigma_i^2)}\right)^2\, d\mu_i(x) \\
&= \dint_{\R^n} \left(\frac{g(x,\sigma_{i+1}^2)}{g(x,\sigma_i^2)}\right)^2 \cdot \frac{g(x,\sigma_i^2)}{G(\sigma_i^2)}\, dx \\
&= \frac{1}{G(\sigma_i^2)} \dint_{\R^n} \frac{g(x,\sigma_{i+1}^2)^2}{g(x,\sigma_i^2)}\, dx \\
&= \frac{1}{G(\sigma_i^2)}\dint_{\R^n} g\left(x,\frac{\sigma_{i+1}^2\sigma_i^2}{2\sigma_i^2-\sigma_{i+1}^2}\right)\, dx \\
&= \frac{G(\frac{\sigma_{i+1}^2\sigma_i^2}{2\sigma_i^2-\sigma_{i+1}^2})}{G(\sigma_i^2)}.
\end{align*}

To bound the number of samples $X$ needed to estimate $Y$ within a target relative error, we will bound $\E(Y^2)/\E(Y)^2$, which is given by
\[
\frac{\E(Y^2)}{\E(Y)^2} = \frac{G(\frac{\sigma_{i+1}^2\sigma_i^2}{2\sigma_i^2-\sigma_{i+1}^2})G(\sigma_i^2)}{G(\sigma_{i+1}^2)^2}.
\]

Then letting $\sigma^2 = \sigma_{i+1}^2$ and $\sigma_i^2 = \sigma^2/(1+\alpha)$, we can further simplify as 
\[
\frac{\E(Y^2)}{\E(Y)^2} = \frac{G\left(\frac{\sigma^2}{1+\alpha}\right)G\left(\frac{\sigma^2}{1-\alpha}\right)}{G(\sigma^2)^2}.
\]

The above $n$-dimensional inequality is difficult to analyze directly. We will reduce it to a simpler $1$-dimensional inequality via localization. Define an exponential needle $E = (a,b,\gamma)$ as a segment $[a,b] \subseteq \R^n$ and $\gamma \in \R$ corresponding to the weight function $e^{\gamma t}$ applied the segment $[a,b]$. The integral of an $n$-dimensional function $f$ over this one dimensional needle is 
\[
\int_E f = \int_0^{|b-a|} f(a + tu) e^{\gamma t}\, dt \qquad \qquad \mbox{ where } \quad u = \frac{b-a}{|b-a|}.
\]

We use the following theorem from~\cite{KLS95}.

\begin{theorem}(\cite{KLS95}) \label{thm:exp-needles}
Let $f_1, f_2, f_3, f_3$ be four nonnegative continuous functions defined on $\R^n$, and $\alpha, \beta > 0$. Then, the following are equivalent:
\begin{enumerate}
\item For every logconcave function $F$ defined on $\R^n$ with compact support,
\[
\left(\int_{\R^n}F(t) f_1(t)\, dt \right)^\alpha 
\left(\int_{\R^n}F(t) f_2(t)\, dt \right)^\beta \le  
\left(\int_{\R^n}F(t) f_3(t)\, dt \right)^\alpha 
\left(\int_{\R^n}F(t) f_4(t)\, dt \right)^\beta
\]
\item For every exponential needle $E$, 
\[
\left(\int_E f_1\right)^\alpha
\left(\int_E f_2\right)^\beta  \le 
\left(\int_E f_3\right)^\alpha
\left(\int_E f_4\right)^\beta
\]
\end{enumerate}
\end{theorem}

A crucial aspect of our proof is that we can restrict the support of our target logconcave function $f$, which then allows us to consider a restricted family of needles. Recall that we assumed $\E_f(\|X\|^2) = R^2$. Set $R_1 = 2 R\cdot\log(1/\eps)$. By the following lemma from~\cite{Lovasz2007}, if we restrict the support of $f$ to be $R_1 \cdot B_n$, we only lose an $\eps/2$ fraction of the mass.

\begin{lemma}\cite{Lovasz2007}
Let $X \in \R^n$ be a random point from a logconcave distribution with $\E(X^2) = R^2$. Then for any $t > 1, \Pr(\|X\| > tR) < \e{-R+1}$. 
\end{lemma}

We can now reduce the desired inequality to a simpler form of exponential needles, which are restricted to lie in the interval $[-R_1, R_1]$. 

\begin{lemma}\label{lem:exp-needle-simplify}
If for all intervals $[\ell,u] \subseteq [-R_1,R_1]$ and $\gamma>0$, 
\[
\frac{\dint_\ell^u \e{\gamma t} \e{-\frac{t^2(1+\alpha)}{2\sigma^2}}\, dt \cdot \dint_\ell^u \e{\gamma t} \e{-\frac{t^2(1-\alpha)}{2\sigma^2}} \, dt}{\left ( \dint_\ell^u \e{\gamma t} \e{-\frac{t^2}{2\sigma^2}}\, dt \right) ^2} \leq c,
\]
then for all logconcave functions $f$ defined on $\R^n$ whose support is a compact subset of $R_1\cdot B_n$,
\[
\frac{G(\frac{\sigma^2}{1+\alpha})G(\frac{\sigma^2}{1-\alpha})}{G(\sigma^2)^2} \le c.
\]
\end{lemma}
\begin{proof}
Applying Theorem~\ref{thm:exp-needles} and setting $f_1(x) = g(\sigma^2/(1+\alpha),x), $\\
$f_2(x) = g(\sigma^2/(1-\alpha),x), f_3(x) = f_4(x) = \sqrt{c} \cdot g(\sigma^2,x), \beta=\gamma=1$, we have that 
\[
\frac{G(\frac{\sigma^2}{1+\alpha})G(\frac{\sigma^2}{1-\alpha})}{G(\sigma^2)^2} \le c
\]
if and only if for all exponential needles $E \subseteq \R^n$, 
\[
\frac{\dint_E g\left(\frac{\sigma^2}{1+\alpha},x\right) \, dx \dint_E g\left(\frac{\sigma^2}{1-\alpha},x\right) \, dx}{\left(\dint_E g\left(\sigma^2,x\right) \, dx\right)^2} \le c.
\]

To prove the lemma, we will show that we can reduce the inequality for an arbitrary exponential needle $E \subseteq \R^n$ to the simpler form. $E$ is defined by an interval $\mathcal{I}$ in $\R^n$ and an arbitrary exponential function $\e{\gamma t}$ on $\mathcal{I}$. Define $z$ as the closest distance from the origin to the extension of the $\mathcal{I}$ in both directions. Parameterize the interval $\mathcal{I}$ in terms of $t$, where $t=0$ gives the closest point along the extension of $\mathcal{I}$ to the origin (note $t=0$ does not necessarily have to be on $\mathcal{I}$). Also define the minimum and maximum values of $t$ on $\mathcal{I}$ as $\ell$ and $u$ respectively. We can assume that $-R_1 \le \ell \le u \le R_1$ since $f$ is $0$ outside of $R_1 \cdot B_n$. We then have that 
\begin{align*}
\int_E g(\sigma^2,x)\, dx &= \int_\ell^u \e{\gamma t} \e{-\frac{t^2+z^2}{2\sigma^2}}) \, dt \\
&=\e{-\frac{z^2}{2\sigma^2}}\cdot \int_\ell^u \e{\gamma t} \e{-\frac{t^2}{2\sigma^2}} \, dt.
\end{align*}

Note that in the integral ratio, the terms with $z$ cancel out since 
\[
\e{-\frac{(1+\alpha)z^2}{2\sigma^2}-\frac{(1-\alpha)z^2}{2\sigma^2} + \frac{2z^2}{2\sigma^2}} = 1,
\]
which then proves the lemma.
\end{proof}

Before bounding the desired inequality, we first prove the following two helper lemmas.

\begin{lemma}\label{lem:4th-moment}
Let $X$ be a random variable with $\E(X^4)<\infty$ and $a \le X \le b$. Then, 
\[
\E(X^4) - \E(X^2)^2 \le 4\max\{a^2,b^2\} \sf{Var}(X).
\]
\end{lemma}
\begin{proof}
Let $Y$ be an independent random variable drawn from the same distribution as $X$. Then,
\begin{align*}
2\Var(X^2) &= \Var(X^2) + \Var(Y^2) \\
&= \E(X^4) - \E(X^2)^2 + \E(Y^4) -\E(Y^2)^2 \\
&= \E(X^4) - 2\E(X^2)\E(Y^2) + \E(Y^4) \\
&= \E\left((X^2-Y^2)^2\right) \\
&= \E\left((X+Y)^2(X-Y)^2\right) \\
&\le 4\max\{a^2,b^2\}\E\left((X-Y)^2\right) \\
&= 4\max\{a^2,b^2\} \E\left(X^2-2XY+Y^2\right) \\
&= 8\max\{a^2,b^2\} \Var(X).
\end{align*}
\end{proof}

\begin{lemma}\label{lem:v-deriv}
Let $[\ell,u] \subseteq [-R_1,R_1]$ and
\[
v(x) = \frac{\int_\ell^u t^2 \e{\gamma t} \e{-\frac{t^2x}{2\sigma^2}}\, dt}{\int_\ell^u \e{\gamma t} \e{-\frac{t^2x}{2\sigma^2}}\, dt}.
\]
Then, $v'(x) \ge -2R_1^2/x$.
\end{lemma}
\begin{proof}
For convenience, define
\[
s(x,t) = \e{\gamma t} \e{-\frac{t^2x}{2\sigma^2}}.
\]
We have that 
\[
v'(x) = \left(\frac{1}{2\sigma^2}\right)\cdot \frac{\left(\dint_\ell^u t^2 s(x,t)\, dt \right)^2-\dint_\ell^u s(x,t) \, dt \dint_\ell^u t^4 s(x,t) \, dt}{\left(\dint_\ell^u s(x,t)\, dt\right)^2}.
\]
Observe that the above quantity is the difference of moments of a truncated Gaussian distribution. We then have that 
\begin{align*}
v'(x) &= \left(\frac{1}{2\sigma^2}\right)\cdot \left(\E(X^2)^2 - \E(X^4)\right)  &\mbox{ where } X \sim \mathcal{N}(\frac{\gamma\sigma^2}{x},\frac{\sigma^2}{x})\Big| \ell \le X \le u \\
&\ge -\frac{2R_1^2}{\sigma^2}\cdot \Var(X) &\mbox{ by Lemma~\ref{lem:4th-moment}} \\
&\ge -\frac{2R_1^2}{\sigma^2} \cdot \frac{\sigma^2}{x} &\mbox{by Theorem~\ref{thm:Brascamp-Lieb}} \\
&=-\frac{2R_1^2}{x}.
\end{align*}

\end{proof}

The following lemma now proves the variance bound.

\begin{lemma}\label{lem:exp-bound}
Let $[\ell,u] \subseteq [-R_1,R_1]$ and $\alpha \le 1/2$. Then
\[
\frac{\dint_\ell^u \e{\gamma t} \e{-\frac{t^2(1+\alpha)}{2\sigma^2}}\, dt \cdot \dint_\ell^u \e{\gamma t} \e{-\frac{t^2(1-\alpha)}{2\sigma^2}} \, dt}{\left ( \dint_\ell^u \e{\gamma t} \e{-\frac{t^2}{2\sigma^2}}\, dt \right) ^2} \leq \e{2 \cdot \frac{R_1^2\alpha^2}{\sigma^2}}.
\]
\end{lemma}
\begin{proof}

Again for convenience, define
\[
s(x,t) = \e{\gamma t} \e{-\frac{t^2x}{2\sigma^2}}.
\]

Define 
\[
h(\alpha) := \frac{\dint_\ell^u s(1+\alpha,t)\, dt \cdot \dint_\ell^u s(1-\alpha,t) \, dt}{\left ( \dint_\ell^u s(1,t) \, dt \right) ^2}.
\]

Note that the lemma is equivalent to bounding $h(\alpha)$. We first prove the following claim, from which the lemma will easily follow.

\begin{claim}\label{claim:deriv}
For $\alpha \le 1/2$, 
\[
h'(\alpha) \le \frac{4 \cdot \alpha R^2  \cdot h(\alpha)}{\sigma^2}.
\]
\end{claim}
\begin{proof}
First, observe that 

\begin{align*}
\frac{\partial}{\partial \alpha}\Big(s(1+\alpha,t)\Big) & = \frac{\partial}{\partial \alpha} \left(\int_\ell^u \e{\gamma t} \e{-\frac{t^2(1+\alpha)}{2\sigma^2}}\, dt\right) \\ \\
&=\frac{-1}{2\sigma^2} \left(\int_\ell^u t^2 \e{\gamma t} \e{-\frac{t^2(1+\alpha)}{2\sigma^2}}\, dt \right)
\end{align*}
and similarly 
\[
\frac{\partial}{\partial \alpha}\Big(s(1-\alpha,t)\Big)=\frac{1}{2\sigma^2} \left(\int_\ell^u t^2 \e{\gamma t} \e{-\frac{t^2(1-\alpha)}{2\sigma^2}}\, dt \right).
\]
Then taking the derivative of $h(\alpha)$ with respect to $\alpha$ gives

\begin{align*}
\frac{\partial}{\partial \alpha} \left(h(\alpha)\right) &=  \frac{\partial}{\partial \alpha} \left(\frac{\dint_\ell^u s(1+\alpha,t)\, dt \cdot \dint_\ell^u s(1-\alpha,t) \, dt}{\left ( \dint_\ell^u s(1,t) \, dt \right) ^2}\right)\\ \\
&= \frac{1}{2\sigma^2} \cdot \frac{\dint_\ell^u s(1+\alpha,t)\, dt \cdot  \dint_\ell^u t^2 s(1-\alpha,t)\, dt - \dint_\ell^u s(1-\alpha,t)\, dt \cdot \dint_\ell^u t^2 s(1+\alpha,t)\, dt}{\left ( \dint_\ell^u s(1,t) \, dt \right) ^2}. \\
\end{align*}

We now have that 
\[
\frac{h'(\alpha)}{h(\alpha)} = \frac{1}{2\sigma^2} \cdot \left(
\frac{\int_\ell^u t^2 \e{\gamma t} \e{-\frac{t^2(1-\alpha)}{2\sigma^2}}\, dt}
{\int_\ell^u \e{\gamma t} \e{-\frac{t^2(1-\alpha)}{2\sigma^2}}\, dt} -
\frac{\int_\ell^u t^2 \e{\gamma t} \e{-\frac{t^2(1+\alpha)}{2\sigma^2}}\, dt}
{\int_\ell^u \e{\gamma t} \e{-\frac{t^2(1+\alpha)}{2\sigma^2}}\, dt}
\right).
\]
Let
\[
v(x) = \frac{\int_\ell^u t^2 \e{\gamma t} \e{-\frac{t^2nx}{2\sigma^2}}\, dt}{\int_\ell^u \e{\gamma t} \e{-\frac{t^2nx}{2\sigma^2}}\, dt}.
\]
We then have that
\begin{align*}
\frac{h'(\alpha)}{h(\alpha)} &= \frac{1}{2\sigma^2} \left(\left(v(1+\alpha) + \dint_{1+\alpha}^{1-\alpha} v'(x)\, dx\right)  - v(1+\alpha)\right) \\
&= -\frac{1}{2\sigma^2} \left(\dint_{1-\alpha}^{1+\alpha} v'(x)\, dx \right) \\
&\le \frac{1}{2\sigma^2} \left( \dint_{1-\alpha}^{1+\alpha} \frac{2R_1^2}{x}\, dx \right) &\mbox{ by Lemma~\ref{lem:v-deriv}} \\
&=\frac{R_1^2}{2\sigma^2} \left( -\frac{1}{x^2}\bigg|_{1-\alpha}^{1+\alpha} \right) \\
&=\frac{R_1^2}{2\sigma^2} \left( \frac{1}{(1-\alpha)^2} - \frac{1}{(1+\alpha)^2}\right) \\
&=\frac{R_1^2}{2\sigma^2} \left(\frac{4\alpha}{(1-\alpha^2)^2}\right) \\
&\le\frac{4R_1^2\alpha}{\sigma^2}.
\end{align*}

\end{proof}

By Claim~\ref{claim:deriv}, we then have a bound on $h(\alpha)$ as follows:

\begin{align*}
\ln h(\alpha) &= \ln h(0) + \int_0^\alpha \frac{d}{dx} \left(\ln h(x)\right)\, dx \\
&=\ln(1) + \int_0^\alpha \frac{h'(x)}{h(x)}\, dx \\
&\le \int_0^\alpha \frac{4R_1^2 x  }{\sigma^2}\, dx \\
&=\frac{2R_1^2x^2}{\sigma^2} \Big |_0^\alpha \\
&= \frac{2R_1^2\alpha^2}{\sigma^2},
\end{align*} 

and thus 
\[
h(\alpha) \le \e{\frac{2R_1^2\alpha^2}{\sigma^2}}.
\]

\end{proof}
\begin{lemma}\label{lem:lc-variance-bound}
Suppose a logconcave function $f:\R^n\rightarrow \R$ has support contained in $R_1 \cdot B_n$. Let $g(x,\sigma^2) = f(x)\e{-\|x\|^2/(2\sigma^2)}$. Let $X$ be drawn from a distribution proportional to $g(x,\sigma^2)$ and $Y=g(X,\sigma^2(1+\alpha))/g(X,\sigma^2)$. Then for $\alpha \le 1/2$,
\[
\frac{\E(Y^2)}{\E(Y)^2} \le \e{\frac{2R_1^2\alpha^2}{\sigma^2}}.
\]
\end{lemma}
\begin{proof}
Follows immediately from Lemma~\ref{lem:exp-needle-simplify} and Lemma~\ref{lem:exp-bound}.
\end{proof}

The bound in Lemma~\ref{lem:variance-bound} then follows by applying Lemma~\ref{lem:lc-variance-bound} with the indicator function of a convex body.

\subsection{Proof of the main theorem}

In this section, we prove Theorem~\ref{thm:volume} by analyzing the runtime of the algorithm in Figure~\ref{fig:volume-alg} and also showing that the volume estimate it computes is accurate. 

The following lemma says that the beginning and ending $\sigma^2$ for the algorithm are sufficient.
\begin{lemma}\label{lem:starting-gauss}
If $\sigma^2 \le 1/(n+\sqrt{8n\ln(1/\eps)})$ and $B_n \subseteq K$, then
\[
\int_K \e{-\frac{-\|x\|^2}{2\sigma^2}}\, dx \ge (1-\eps) \int_{\R^n}\e{-\frac{-\|x\|^2}{2\sigma^2}}\, dx.
\]
Let $K\subseteq C\sqrt{n}B_n$, $f_i(x) = \exp(-\frac{\|x\|^2}{2\sigma_i^2})$, $\sigma_i^2 \ge C^2n$, and $\sigma_{i+1}^2 = \infty$. Then for $X$ drawn from distribution proportional to $f_i \cap K$ and $Y = f_{i+1}(X)/f_i(X)$,
\[
\frac{\E(Y^2)}{\E(Y)^2} \le e^2.
\]

\end{lemma}
\ifdraft
\else
\begin{proof}
The proof that the starting Gaussian has most of its measure in the unit ball is in \cite{CV2014}. To show that $\sigma_i^2 \ge C^2n$ is sufficient to switch to the uniform distribution, observe that $f_i(X) \ge 1/e$ since $\|X\|^2 \le C^2n$ and thus $Y \le e$. Also note $Y \ge 1$. Therefore $\E(Y^2)/\E(Y)^2 \le e^2$.
\end{proof}
\fi

We now bound the variance when $\sigma^2$ is small. First, we will need the following lemma that is proved in \cite{LV2}.

\begin{lemma} \label{lem:z-logconcave}
Let $K \subseteq \mathbb{R}^n$ be a convex body and $f : K \rightarrow \mathbb{R}$ be a logconcave function. For any $a > 0$, define 

\[
Z(a) = \int_K f(ax) dx.
\]

\noindent
Then $a^nZ(a)$ is a logconcave function of $a$.

\end{lemma}
\begin{lemma}\label{lem:fixed-var-bound}
Assume $n\ge 3$. Let $X$ be a random point in $K$ with density proportional to $f_i(x) = \e{-\frac{\|x\|^2}{2\sigma_i^2}}$, $\sigma_{i+1}^2 = \sigma_i^2 (1+1/n)$, and $Y = f_{i+1}(X)/f_i(X)$. Then,
\[
\frac{\E(Y^2)}{\E(Y)^2} < 1+\frac{2}{n}.
\]
\end{lemma}
\begin{proof}
We have
\[
\frac{\E(Y^2)}{\E(Y)^2} = \frac{\dint_K \e{-\frac{\|x\|^2(1-\alpha)}{2\sigma^2}}\, dx \dint_K \e{-\frac{\|x\|^2(1+\alpha)}{2\sigma^2}}\, dx}{\left(\dint_K \e{-\frac{\|x\|^2}{2\sigma^2}}\, dx \right)^2}.
\]

By Lemma \ref{lem:z-logconcave}, the function $z(a) = a^{n+1} \int_K \e{-a\|x\|^2/2}\, dx$ is logconcave, and thus
\[
z\left(\frac{1-\alpha}{\sigma^2}\right)z\left(\frac{1+\alpha}{\sigma^2}\right) \leq z\left(\frac{1}{\sigma^2}\right)^2.
\]
Therefore
\[
\dint_K \e{-\frac{\|x\|^2(1-\alpha)}{2\sigma^2}}\, dx \dint_K \e{-\frac{\|x\|^2(1+\alpha)}{2\sigma^2}}\, dx \leq \left(\frac{1}{1-\alpha^2}\right)^{n+1}\left(\dint_K \e{-\frac{\|x\|^2}{2\sigma^2}}\, dx\right)^2.
\]

Setting $\alpha = 1/n$, we have that
\[
\frac{\E(Y^2)}{\E(Y)^2} \leq \left(\frac{1}{1-1/n^2}\right)^{n+1} = \left(1+\frac{1}{n^2-1}\right)^{n+1} \le \e{\frac{1}{n-1}} \le 1+\frac{2}{n}.
\]

\end{proof}

We now show that the volume estimate computed in Algorithm~\ref{fig:volume-alg} is accurate. Define $R_i$ as the $i$-th integral ratio, i.e.
\[
R_i := \frac{F(\sigma_{i+1}^2)}{F(\sigma_i^2)} = \frac{\int_K \e{-\|x\|^2/(2\sigma_{i+1}^2)}\, dx}{\int_K \e{-\|x\|^2/(2\sigma_i^2)}\, dx},
\]
and let $W_i$ denote the estimate of the algorithm for $R_i$.

For two random variables $X,Y$, we will measure their independence by the following:
\[
\mu (X,Y) = \sup_{A,B} |P(X \in A, Y \in B) - P(X \in A)P(Y \in B)|,
\]
where $A,B$ range over measurable subsets of the ranges of $X,Y$.

We will give an argument similar to \cite{LV2}, and use the following lemmas that were proved there.

\begin{lemma}\cite{LV2}\label{lem:fn-indep}
If $f$ and $g$ are two measurable functions, then 
\[
\mu (f(X), g(Y)) \leq \mu (X,Y).
\]
\end{lemma}

\begin{lemma}\cite{LV2}\label{lem:cov-bd}
Let $X,Y$ be random variables such that $0 \leq X \leq a$ and $0 \leq Y \leq b$. Then
\[
|\E(XY)-\E(X)E(Y)| \leq ab\mu(X,Y). 
\]
\end{lemma}

\begin{lemma}\cite{LV2}\label{lem:exp-bd}
Let $X\geq 0$ be a random variable, $a>0$, and $X' = \min (X,a)$. Then
\[
\E(X') \geq E(X) - \frac{\E(X^2)}{4a}.
\]
\end{lemma}

\begin{lemma}\label{lem:accuracy}
With probability at least $4/5$, 
\[
(1-\eps)R_1\ldots R_m \le W_1 \ldots W_m \le (1+\eps) R_1 \ldots R_m.
\]
\end{lemma}
\begin{proof}
Let $(X_0^i, X_1^i, X_2^i, \ldots, X_k^i)$ be the sequence of sample points for the $i$th volume phase. The distribution of each $X_i$ is approximately the correct distribution, but slightly off based on the error parameter $\nu$ in each phase that bounds the total variation distance. We will define new random variables $\overline X_j^i$ that have the correct distribution for each phase.

Note that $X_j^0$ would be sampled from the exact distribution, and then rejected if outside of $K$. Therefore $\Pr(X_j^0 = \overline X_j^0)=1$. Suppose that the total number of sample points throughout the algorithm is $t$. Using induction and the definition of total variation distance, we see that 
\begin{equation}\label{eqn:exact-chance}
\Pr(X_i^j = \overline X_i^j, \forall i,j) \ge 1-t\nu.
\end{equation}
Let
\[
\quad Y_j^i = \frac{\e{-\frac{\|X_j^i\|^2}{2\sigma_{i+1}^2}}}{\e{-\frac{\|X_j^i\|^2}{2\sigma_i^2}}}
\quad \mbox{ and } \quad 
\overline W_i = \frac{1}{k_i} \sum_{j=1}^{k_i}  Y_j^i.
\]

Note that for a fixed $i$, all of the $Y_j^i$ have the same expectation since they are from the exact distribution, and it is equal to $\E(\overline W_i)$. Suppose that we have $\E((Y_j^i)^2) \le c_i \E(Y_j^i)^2$. Then
\begin{align}
\E(\overline W_i^2) &= \frac{1}{k_i^2} \left(\sum_{j=1}^{k_i} \E((Y_j^i)^2)  + k_i(k_i-1)R_i^2\right) \nonumber \\
&\le \left(1+\frac{c_i-1}{k_i}\right) \cdot \E( \overline W_i)^2.\label{eq:var-bound}
\end{align}

The following claim bounds the variance of our ratio estimater under a faster cooling rate; combined with Lemma~\ref{lem:fixed-var-bound}, we have a bounds on the variance throughout our algorithm. It follows from Lemma~\ref{lem:variance-bound}.

\begin{claim}\label{claim:fast-var-bound}
Suppose that $K \subseteq C\sqrt{n}B_n$ and let $\alpha = \sigma^2/(2C^2n)$. Then,
\[
\frac{\E\left(( Y^i )^2\right)}{\E( Y^i)^2} < 1+\frac{\sigma^2}{C^2n}.
\]
\end{claim}

Suppose that we had independence between samples and consider bounding the cumulative error for all phases of the algorithm. When $\sigma^2\le 1$, we can bound the number of phases for the first part as $m_1 \le 2n\log 4n$. When $\sigma^2 > 1$, we will analyze the phases in chunks, where a chunk is the set of phases until $\sigma^2$ doubles. Note that the number of phases in a chunk starting with variance $\sigma^2$ is at most $2C^2n/\sigma^2$. Also there are at most $\log(C^2 n)$ chunks. Observe that for a single chunk with starting variance $\sigma^2$, where $i,j$ are the starting and ending phases of the chunk, we have
\[
\frac{\E(\overline{W}_i^2 \ldots \overline{W}_j^2)}{R_i^2 \ldots R_j^2} \le \left(1+\frac{2\sigma^2}{kC^2n}\right)^{2C^2n/\sigma^2} \le \left(1+\frac{5}{k}\right).
\]
Then, there is one final phases when we switch to the uniform distribution, which has variance at most $1+e^2$ by Lemma~\ref{lem:starting-gauss}. 

%Since $\sigma^2 \ge 1$, the number of phases for the second part is at most $2C^2n \log C^2n$. We can then bound the total number of phases of the entire algorithm as $m \le 5n \log C^2n$. 

Let $m$ denote the total number of phases. If we had independence between samples, then we can use Lemma~\ref{lem:fixed-var-bound} and Claim~\ref{claim:fast-var-bound} with Chebyshev's inequality to bound the probability of failure\ifdraft
. The remainder of the proof is in the extended version of the paper.
\else
: 

\begin{align*}
\Pr\left(\frac{|\overline W_1 \ldots \overline W_{m} - R_1 \ldots R_{m}|}{R_1 \ldots R_{m}} \ge \frac{\eps}{2} \right) &\le \frac{4\Var(\overline W_1 \ldots \overline W_{m})}{\eps^2 R_1^2 \ldots R_{m}^2}\\
&= \frac{4}{\eps^2}\left( \frac{\E(\overline W_1 ^2 \ldots \overline W_{m}^2)}{R_1^2\ldots R_{m}^2}-1\right)\\
&\le \frac{4}{\eps^2} \left(\left(1+\frac{2}{kn}\right)^{m_1}\left(1+ \frac{5}{k}\right)^{\log C^2n} \left(1+\frac{e^2}{k}\right) -1\right)\\
%&\le \frac{4}{\eps^2} \left( \e{\frac{2m_1}{kn} + \frac{5}{k} + \ldots +\frac{2}{k}} - 1 \right)\\
&\le \frac{4}{\eps^2} \left(\e{\frac{2m_1}{kn} + \frac{5\log(C^2n)}{k}+\frac{e^2}{k}}-1\right)\\
&\le \frac{4}{\eps^2} \left(\e{\frac{\eps^2}{50}}-1\right) \\
&\le\frac{4}{\eps^2}\left(\left(1+\frac{\eps^2}{40}\right)-1\right) \\
&= \frac{1}{10}
\end{align*}

However, subsequent samples are dependent, and we must carefully bound the dependence. The analysis is somewhat involved, but will follow essentially the sample template as in \cite{LV2,CV2014} which utilizes the following lemma to bound dependence between subsequent samples, where $\nu$ is the target total variation distance for each sample point. For convenience,  denote the entire sequence of $t$ samples points used in the algorithm as $(Z_0, Z_1, \ldots, Z_{t-1})$. 

\begin{lemma}\label{lem:delta-ind}

(a) For $0 \leq i < t$, the random variables $Z_i$ and $Z_{i+1}$ are $\nu$-independent, and the random variables $\overline{Z}_i$ and $\overline{Z}_{i+1}$ are $(3 \nu)$-independent.

(b) For $0 \leq i < t$, the random variables $(Z_0, \ldots ,Z_i)$ and $Z_{i+1}$ are $(3\nu)$-independent.

(c) For $0 \leq i < m$, the random variables $\overline{W}_1 \ldots \overline{W}_i$ and $\overline{W}_{i+1}$ are $(3km\nu)$-independent.
\end{lemma}
The variables $\overline W_i$ are not bounded, but we will introduce a new set of random variables based on $\overline W_i$ that are bounded so we can later apply Lemma \ref{lem:cov-bd}. Let
\[
\alpha = \frac{\eps^{1/2}}{8(m\mu)^{1/4}},
\]
where $\mu = 3km\nu$. Note that $\alpha$ is much larger than one. Define
\[
V_i = \min\{\overline W_i, \alpha \E (\overline W_i) \}.
\]

It is clear that $\E(V_i) \leq \E(\overline W_i)$, and by Lemma \ref{lem:exp-bd}, we also have
\[
\E(V_i) \geq \E(\overline W_i) - \frac{\E(\overline W_i^2)}{4\alpha \E(\overline W_i)} \geq (1 - \frac{1}{4\alpha}(1 + \frac{7}{k}))\E(\overline W_i) \geq (1-\frac{1}{2\alpha})\E(\overline W_i).
\]

Let $U_0=1$ and define recursively
\[
U_{i+1} = \min\{U_iV_{i+1}, \alpha \E(V_1) \ldots \E(V_{i+1})\}.
\]

We will now show that 
\begin{equation}\label{eqn:u-bound}
(1-\frac{i-1}{\alpha})\E(V_1) \ldots \E(V_i) \leq \E(U_i) \leq (1 + 2\mu \alpha^2 i)\E(V_1) \ldots \E(V_i).
\end{equation}
By Lemma \ref{lem:fn-indep}, the random variables $U_i$ and $V_{i+1}$ are $\mu$-independent, and by Lemma \ref{lem:cov-bd} and since $\alpha \geq 1$, 
\begin{equation}\label{eqn:uv-bound}
|\E(U_iV_{i+1} - \E(U_i)\E(V_{i+1})| \leq \mu \alpha \E(V_1) \ldots \E(V_i)\alpha \E(\overline W_{i+1}) \leq 2\mu\alpha^2 \E(V_1) \ldots \E(V_{i+1}).
\end{equation}

From (\ref{eqn:uv-bound}), we can get the upper bound on $\E(U_{i+1})$ by induction:
\begin{align*}
\E(U_{i+1}) &\leq \E(U_i V_{i+1}) \\
&\leq \E(U_i)\E(V_{i+1}) + 2 \mu \alpha^2 \E(V_1) \ldots \E(V_{i+1})\\
&\leq (1+2\mu \alpha^2(i+1))\E(V_1) \ldots \E(V_{i+1}). \numberthis \label{eqn:u-upperbound}
\end{align*}

Similarly, 
\begin{align}
&\E(U_i^2) \leq (1+2\mu \alpha^4 i) \E(V_1^2) \ldots \E(V_i^2)\label{eqn:u2-bound}\\
\text{and } \>\> &\E(U_i^2V_{i+1}^2) \leq (1 + 2\mu \alpha^4 i)\E(V_1^2) \ldots \E(V_{i+1}^2).\label{eqn:uv2-bound}
\end{align}

For the lower bound, we use Lemma \ref{lem:exp-bd} and (\ref{eqn:uv2-bound}) to get:
\begin{align*}
\E(U_{i+1}) &\geq \E(U_iV_{i+1}) - \frac{\E(U_i^2V_{i+1}^2)}{4\alpha \E(V_1) \ldots \E(V_{i+1})}\\
&\geq \E(U_iV_{i+1}) - (1+2\mu \alpha^4 i)\frac{\E(V_1^2) \ldots \E(V_{i+1}^2)}{4\alpha \E(V_1) \ldots \E(V_{i+1})}. \numberthis \label{eqn:u-lowerbound}
\end{align*}

For $\alpha \geq 3k$, we have that
\begin{align*}
\E(V_i^2) \leq \E(\overline{W}_i^2) &\leq (1+\frac{7}{k})\E(\overline{W}_i)^2\\
&\leq (1+\frac{7}{k})\frac{1}{(1-1/(2\alpha))^2}\E(V_i)^2\\
&\leq (1+\frac{7}{k})(1+\frac{1}{2k})\E(V_i)^2\\
&\leq (1+\frac{8}{k})\E(V_i)^2. \numberthis \label{eqn:v2-bound}
\end{align*}

Combining (\ref{eqn:uv-bound}), (\ref{eqn:u-lowerbound}), and (\ref{eqn:v2-bound}), 
\begin{align*}
\E(U_{i+1}) &\geq \E(U_i V_{i+1}) - \frac{1}{4\alpha}(1+2\mu\alpha^4i)(1+\frac{8}{k})^i\E(V_1) \ldots \E(V_{i+1})\\
&\geq \E(U_iV_{i+1})-\frac{1+2\mu\alpha^4i}{2\alpha}\E(V_1) \ldots \E(V_{i+1})\\
&\geq \E(U_i)\E(V_{i+1}) - \frac{1}{\alpha}\E(V_1) \ldots \E(V_{i+1}).
\end{align*}

Then, by induction on $i$,
\begin{equation}\label{eqn:u-lowerbound2}
\E(U_{i+1}) \geq \E(V_1) \ldots \E(V_{i+1}) - \frac{i}{\alpha}\E(V_1) \ldots \E(V_{i+1}).
\end{equation}

Putting (\ref{eqn:u-upperbound}) and (\ref{eqn:u-lowerbound2}) together, we now have a proof of (\ref{eqn:u-bound}). Thus,
\[
\E(U_m) \leq (1 + \frac{\eps}{4}) \E(V_1) \ldots \E(V_m) \leq (1 + \frac{\eps}{4})\E(\overline{W}_1) \ldots \E(\overline{W}_m).
\]

We also have that $\alpha \geq 4m/\eps$ implies 
\[
\E(U_m) \geq (1-\frac{\eps}{4})\E(\overline{W}_1) \ldots \E(\overline{W}_m).
\]

From (\ref{eqn:u2-bound}) and (\ref{eqn:v2-bound}), and the selection of $\alpha$, $\mu$, and the lower bound on $k$, we have that 
\begin{align*}
\E(U_m^2) &\leq (1+2\mu \alpha^4 m)\E(V_1^2) \ldots \E(V_m^2)\\
&\leq (1+2\mu \alpha^4 m)(1+\frac{8}{k})^m\E(V_1)^2 \ldots \E(V_m)^2\\
&\leq (1+2\mu \alpha^4 m)(1+\frac{8}{k})^m\frac{1}{(1-(m-1)/\alpha)^2}\E(U_m)^2\\
&\leq (1 + \frac{\eps^2}{64})E(U_m)^2,
\end{align*}

and hence

\[
\Pr\big (|U_m - \E(U_m)| \leq \frac{\eps}{2}\E(\overline{W}_1) \ldots \E(\overline{W}_m)\big ) \geq 0.9
\]
 by Chebyshev's inequality. Then, applying Markov's inequality,
\[
\Pr(U_{i+1} \neq U_iV_{i+1}) = \Pr \big (U_iV_{i+1} > \alpha \E(V_1) \ldots \E(V_{i+1})\big )\leq\frac{2}{\alpha}
\]
and similarly
\[
\Pr(V_i \neq \overline{W}_i)\leq \frac{1}{\alpha}.
\]

So, with probability at least $1-3k/\alpha$, we have $U_m = \overline{W}_1 \ldots \overline{W}_m$. Also, from (\ref{eqn:exact-chance}), we have that $\overline{W}_1 \ldots \overline{W}_m = W_1 \ldots W_m$ with probability at least $1 - 2km\nu$. Recall that $\E(\overline{W}_1) \ldots \E(\overline{W}_m) = R_1 \ldots R_m$. Therefore, with probability at least $4/5$
\begin{align*}
|W_1 \ldots W_m - R_1\ldots R_m| \leq \frac{\eps}{2}R_1 \ldots R_m,
\end{align*}
which proves the lemma.
\fi
\end{proof}

We can now prove the main theorem.

\begin{proof}(of Theorem~\ref{thm:volume})

We assume that $\eps \ge 2^{-n}$, which only ignores cases which would take exponential time. Then by Lemma~\ref{lem:starting-gauss}, selecting $\sigma_0^2 = 1/(4n)$ implies that all but a negligible amount of volume of the starting Gaussian is contained in $K$. 

Recall that our algorithm only has a bound on the expected number of steps. To account for this, we will run the algorithm $O(1)$ times to obtain a run which takes at most a constant factor of ball walk steps to proper steps, say with probability $1/20$. By Lemma~\ref{lem:accuracy}, the answer returned by the algorithm will be within the target relative error with probability at least $4/5$.  Thus the overall probability of failure is $3/4$. Note that we can boost this probability of failure to $1-p$ by the standard trick of repeating the algorithm $\log 1/p$ times and returning the median. 

We now analyze the runtime of the algorithm in Figure~\ref{fig:volume-alg}. Set $C=R\log(1/\eps)/\sqrt{n}$. Assume that $C \ge 1$ (otherwise arbitrarily increase $C$). When $\sigma^2 \le 1$, using the value of $k$, the mixing time assigned to each phase, and the fact that there are $O(n \log n)$ phases, we see that the total number of ball walk steps taken is $O(n^{2.5}k\log n \log^2(n/\eps)) = O(n^3\log^2 n \log^2(n/\eps) / \eps^2)=O^*(n^3)$. When $\sigma^2 > 1$, the analysis is very similar if we note that the faster cooling rate and fewer number of samples cancels out the slower mixing time of $O^*(\sigma^2 n^2)$. Thus, it follows that the total number of ball walk steps taken is
\[
O\left(\frac{C^2n^3\log^2 n \log^2\frac{n}{\eps}}{\eps^2}\right) = O\left(\frac{R^2n^2}{\eps^2} \cdot \log^2 n \log^2\frac{1}{\eps} \log^2 \frac{n}{\eps}\right) = O^*(R^2n^2).
\] 
\end{proof}

\section{Conclusion}
We make a few concluding remarks:\\
1.  In our algorithm, the complexity of volume computation for well-rounded bodies is essentially the same as the amortized complexity of generating a single uniform random sample --- both are $O^*(n^3)$. This is in contrast to all previous volume algorithms where the amortized complexity of sampling is lower by at least a factor of $n$ compared to the complexity of volume computation.\\
2. It would be interesting to extend our algorithm to integrating any well-rounded logconcave function; we expect this should be possible with essentially the same complexity. The variance of the ratio of integrals computed in each phase, as well as the isoperimetric inequality, are already proven in full generality for all logconcave functions. \\
3. The accelerated cooling schedule used in our algorithm can be seen as a worst-case analysis of the cooling schedule used in a practical algorithm for volume computation \cite{CV13, CV15b}; in the latter, we used an adaptive schedule by empirically estimating the maximum tolerable change in the variance of the Gaussian that keeps the variance of the ratio estimator bounded by a constant. \\
4. An important open question is to find an $O^*(n^3)$ rounding algorithm for arbitrary convex bodies. The current best rounding complexity is $O^*(n^4)$ \cite{LV2}. \\
%and if the rounding complexity could be improved to $O^*(n^3)$, then we would get an $O^*(n^3)$ volume algorithm for all convex bodies. 

\bibliography{acg}
\ifdraft
\else

\end{document}